\documentclass[draft,onecolumn]{IEEEtran}
\IEEEoverridecommandlockouts

\usepackage{cite}
\usepackage{amsmath,amssymb,amsfonts,amsthm}
\usepackage{algorithmic}
\usepackage{graphicx}
\usepackage{subcaption}
\usepackage{textcomp}
\usepackage{xcolor}
\usepackage[bookmarks=false]{hyperref}

\usepackage[a4paper, left=1.8cm, right=1.8cm, bottom=2.3cm, top=2.3cm]{geometry}

\usepackage[square,numbers]{natbib}

\usepackage{tikz}
\usetikzlibrary{positioning}
\usetikzlibrary{automata}
\usetikzlibrary{arrows}
\tikzset{
    node distance=2.5cm,
    every edge/.append style={ 
        draw,
        ->,
        auto, 
        semithick
    },
    every state/.append style={
        minimum size=1.2cm
    },
    initial/.append style={
        initial text={}
    }
}

\DeclareMathOperator*{\argmax}{argmax}
\DeclareMathOperator*{\argmin}{argmin}

\newtheorem{thm}{Theorem}
\newtheorem{coro}{Corollary}
\newtheorem{defn}{Definition}

\newtheorem{lm}{Lemma}

\theoremstyle{remark}
\newtheorem{rem}{Remark}

\definecolor{bluishgreen}{RGB}{0,158,115}
\definecolor{vermillion}{RGB}{213,94,0}
\definecolor{myblue}{RGB}{0,114,178}
\definecolor{myorange}{RGB}{230,159,0}

\allowdisplaybreaks

\begin{document}
\begin{NoHyper}

\title{Memory Constrained Adversarial Hypothesis Testing
\thanks{The paper has been accepted for presentation in part at the 2026 {\em IEEE Information Theory Workshop (ITW 2026)} under the title ``Limited Memory Adversarial Hypothesis Testing.'' Please see the previous publication statement for more details.}
}

\author{\IEEEauthorblockN{Malhar A. Managoli and Vinod M. Prabhakaran}\\
\IEEEauthorblockA{School of Technology and Computer Science \\
Tata Institute of Fundamental Research\\
Mumbai, India\\
{\tt \{malhar.managoli, vinodmp\}@tifr.res.in}}
}

\maketitle

\begin{abstract}
We study adversarial binary hypothesis testing under memory constraints. The test is a time-invariant randomized finite state machine (FSM) with $S$ states. Associated with each hypothesis is a set of distributions. Given the hypothesis, the distribution of each sample is chosen from the set associated with the hypothesis by an adversary who has access to past samples and the history of states of the FSM so far. We obtain upper and lower bounds on the minimax asymptotic error probability that are tight to within one state. These bounds determine the exact error exponent in $S$, and coincide under a weak least-favourability condition. For time-invariant deterministic FSMs, we also extend known achievable error-exponent bounds from the non-adversarial Bernoulli setting to the adversarial case with parameters in two separated intervals.
\end{abstract}

\section{Introduction}
In binary hypothesis testing, one of the most basic statistical inference tasks, we are given two distributions $p_0,p_1$, and independent and identically distributed (i.i.d.) samples drawn according to one of them. The goal is to determine which distribution the samples came from.
Often, $p_0$ and $p_1$ are not exactly known. Composite hypothesis testing considers sets of distributions $\mathcal{P}_0$ and $\mathcal{P}_1$, and samples are drawn i.i.d. according to a fixed (but unknown) distribution from one of the two sets. In classical robust statistics, these subsets $\mathcal{P}_0$ and $\mathcal{P}_1$ are ``neighbourhoods'' of the ideal model distributions $p_0$ and $p_1$, respectively~\cite{huber1965robust,huber1973minimax,levy2009robust}. Fangwei and Shiyi~\cite{fangwei1996hypothesis,fu1998hypothesis} studied a variant in which the distribution may vary from sample to sample (but all distributions belong to the same set, either $\mathcal{P}_0$ or $\mathcal{P}_1$). More recently, Brand\~{a}o et al.~\cite{brandao2020adversarial} considered adversarial hypothesis testing, in which the distribution of each sample may be adversarially chosen (from one of the sets) depending on previous samples. Other forms of robustness have also been studied in the context of hypothesis testing~\cite{JinLai2021adversarial,puranik2022generalized,PanLiTan2023adversarialgame,managoli25error} and, more generally, statistical inference~\cite{diakonikolas2023algorithmic}.

Another practical consideration is that testing algorithms have finite memory. Cover~\cite{cover1969hypothesis} initiated the study of hypothesis testing with memory constraints by modelling the algorithm as a finite-state machine (FSM). If the FSM is allowed to be time-variant, arbitrarily small probabilities of error can be achieved with only three states~\cite{cover1969hypothesis,koplowitz1975necessary}. For binary hypothesis testing, Hellman and Cover~\cite{hellman1970learning} obtained the optimal trade-off between the probability of error and the number of states for randomized time-invariant FSMs. Their result shows that the probability of error for testing $p$ versus $q$ using an $S$-state randomized time-invariant FSM behaves as $\Theta(\gamma_\mathrm{HC}^{-S/2}(p,q))$, where
\begin{align}
    \gamma_\mathrm{HC}(p,q)&:=\max_{x_+\in\mathcal{X}}\frac{q(x_+)}{p(x_+)}\max_{x_-\in\mathcal{X}}\frac{p(x_-)}{q(x_-)}. \label{eq:intro-gHC}
\end{align}
Lakshmanan and Chandrasekaran \cite{lakshmanan1979compound} extended these results to composite hypothesis testing. More recently, Berg, Ordentlich, and Shayevitz~\cite{berg2020binary} obtained upper and lower bounds on this trade-off for time-invariant deterministic FSMs. Memory-constrained algorithms for statistical inference have received significant recent attention; see~\cite{berg2024statistical} for an excellent survey.

In this work, we study adversarial hypothesis testing using memory-constrained algorithms. Like Hellman and Cover~\cite{hellman1970learning}, we consider time-invariant randomized FSMs. The adversary chooses the distribution of each sample based on previous samples as in Brand\~{a}o et al.~\cite{brandao2020adversarial}. Furthermore, we allow the adversary's choice to also depend on the history of states of the FSM so far. The adversary's choices of distributions are confined to one of two sets, $\mathcal{P}_0$ or $\mathcal{P}_1$, depending on whether the hypothesis is $H=0$ or $H=1$, respectively. We obtain upper and lower bounds on the optimal minimax asymptotic error probability that are tight to within one state: our achievability result with $S+1$ states has a probability of error which we show is no larger than the optimum achievable with $S$ states. Thus, these bounds determine the exact error exponent in $S$. This exponent turns out to be equal to the minimum Hellman–Cover exponent over pairs $p\in\mathcal{P}_0$ and $q\in\mathcal{P}_1$. Our bounds coincide whenever a least favourable distribution pair exists for $\mathcal{P}_0$ and $\mathcal{P}_1$; in fact, we identify a weaker form of this condition which suffices.

We show our achievability using an FSM inspired by that of Hellman and Cover. Their FSM, with states $\{1,\ldots,S\}$, for testing $p$ versus $q$ distributed on $\mathcal{X}$, is organized as a {\em birth-and-death chain} that makes transitions to the right (resp. left) only on observing the letter $x\in\mathcal{X}$ that maximizes (resp. minimizes) the likelihood ratio $q(x)/p(x)$ as in \eqref{eq:intro-gHC}; see Section~\ref{sec:hellman-cover} for a brief description. 
Arguably, this letter provides the best evidence for $H=1$ (resp. $H=0$). In the adversarial setting, no one letter may provide good evidence under all possible distributions. We assign weights to the letters and these weights determine the probability with which our FSM transitions to the right/left. 
We propose the following generalization of \eqref{eq:intro-gHC}.
\begin{align}\gamma:= \sup_{\substack{f_+,f_-}} \; \min_{\substack{p\in\mathcal{P}_0}, \, q\in\mathcal{P}_1} \frac{\mathbb{E}_q[f_+(X)]\mathbb{E}_p[f_-(X)]}{\mathbb{E}_q[f_-(X)]\mathbb{E}_p[f_+(X)]},\label{eq:intro-g}\end{align}
where the $\sup$ is over non-negative functions $f_+$ and $f_-$ defined on $\mathcal{X}$ s.t. $f_+(x)+f_-(x)\leq 1$, $x\in\mathcal{X}$. On observing $x$, our FSM transitions to the right (resp. left) with probabilities $f_+(x)$ (resp. $f_-(x)$; see Section~\ref{sec:achievability} for details.
While Hellman and Cover analyse the stationary distribution of their FSM under $p$ and $q$, we handle the adversarial setting by analysing limiting occupation measures, which satisfy stationary balance equations even when the state distribution does not converge. Our achievability implies that the probability of error satisfies $P_e(S)=O(\gamma^{-S/2})$.
To show our converse, we demonstrate a minimax property for  $\gamma$ and relate it to the Hellman-Cover quantity in \eqref{eq:intro-gHC} minimized over a pair of distributions $p\in\mathcal{P}_0$, $q\in\mathcal{P}_1$. The objective in \eqref{eq:intro-g} need not be jointly quasi-convex in $(p,q)$ or jointly quasi-concave in $(f_+,f_-)$, so Sion's minimax theorem~\cite{sion1958general} does not apply directly. We instead exploit the multiplicative structure of the objective to justify successive applications of the theorem, obtaining
\[ \gamma = \min_{\substack{p\in\mathcal{P}_0}, \, q\in\mathcal{P}_1} \gamma_{\mathrm{HC}}(p,q).\]
Since the minimizing $p\in\mathcal{P}_0$ and $q \in\mathcal{P}_1$ above are valid choices for the adversary, the probability of error behaves as $\Theta(\gamma^{-S/2})$; see Theorem~\ref{thm:error_characterisation} for the precise statement we show.

We also consider time-invariant deterministic FSMs for Bernoulli observations. When the adversary chooses the Bernoulli parameter from one of two separated intervals, we extend the achievable error-exponent bounds from the non-adversarial setting~\cite{shubert1973testing,berg2020binary} to the adversarial setting. Our proof uses a coupling argument exploiting the monotonicity of the corresponding FSMs used by \cite{shubert1973testing,berg2020binary}.

The rest of the paper is organised as follows: In Section~\ref{sec:prelims} we formally state the problem and introduce some relevant concepts.
We state our main results in Section~\ref{sec:results}.
In Sections~\ref{sec:achievability} and~\ref{sec:converse}, we prove our achievability and converse results, respectively. We discuss two illustrative examples in Section~\ref{sec:examples}.
Deterministic FSMs are discussed in Section~\ref{sec:deterministic}.
We conclude with a discussion of future work in Section~\ref{sec:discussion}.

\section{Problem Statement and Preliminaries}
\label{sec:prelims}

Consider a finite alphabet $\mathcal{X}$ and two closed convex\footnote{As will become clear, assuming these sets are convex amounts to allowing the adversary to randomize.} sets of distributions $\mathcal{P}_0$ and $\mathcal{P}_1$ on $\mathcal{X}$ 
corresponding to hypotheses $H=0$ and $H=1$, respectively. 
The memory-constrained detector is a {randomized}, {time-invariant}, {discrete-time}, {finite-state machine} (FSM) on $S$ states with state space $[S]=\{1,2,\ldots,S\}$, input space $\mathcal{X}$, and output space $\{0,1\}$, specified by a {\em state-transition function} $\pi(t|s,x), t,s\in[S],x\in\mathcal{X}$, a {\em decision function} $\hat{H}:[S]\rightarrow\{0,1\}$ and an initial state $R_0\in[S]$.
The state-transition function is non-negative and $\sum_{t} \pi(t|s,x)=1, s\in[S], x\in\mathcal{X}$.
Let $R_n,n\in\mathbb{N}$ denote the detector's state at the end of time step $n$.
During time step $n$, the detector receives a sample $X_n\in\mathcal{X}$ and samples its new state $R_{n}$ according to the distribution $\pi(\cdot|R_{n-1},X_n)$.
The detector's output at time $n$ is $\hat{H}(R_n)$.
Under hypothesis $H=h$, $h\in\{0,1\}$, the sample $X_n$ is drawn from a distribution $p_n$ chosen adversarially from $\mathcal{P}_h$.
When choosing $p_n$, the adversary, who knows the hypothesis, has access\footnote{While our achievability analysis handles such an adversary, to show our converse, we only need to assume a weaker adversary who has no access to either the past samples or the history of the FSM's states.} to past samples $X_1^{n-1}$ and the past sequence of detector states $R_0^{n-1}$.
Thus, an adversarial strategy $\mathcal{A}_h$, under hypothesis $H=h$, consists of a sequence of functions\footnote{Without loss of generality, these are deterministic functions since $\mathcal{P}_0$ and $\mathcal{P}_1$  are convex.} $\{p_n:{\mathcal{X}}^{n-1}\times{[S]}^{n}\rightarrow\mathcal{P}_h\}_{n\in\mathbb{N}\backslash\{0\}}$.

The performance of the detector $(\pi,\hat{H})$, against a pair of adversarial strategies $\mathcal{A}=(\mathcal{A}_0,\mathcal{A}_1)$, is measured by the asymptotic average error probability:
\begin{align*}
    P_e(\pi,\hat{H};\mathcal{A}):=\max_{h\in\{0,1\}} \limsup_{n\rightarrow\infty}\frac{1}{n} \sum_{i\in[n]} \mathbb{P}_{H=h}^{\pi,\mathcal{A}_h}[\hat{H}(R_i)\ne h],
\end{align*}
where $\mathbb{P}_{H=h}^{\pi,\mathcal{A}_h}$ denotes the probability under hypothesis $H=h$ and we suppress the dependence of $P_e$ on the initial state $R_0$.
We are interested in studying the asymptotic minimax risk:
\begin{align}
    P_e^*(S)&:=\inf_{\pi,\hat{H}} \; \sup_\mathcal{A}P_e(\pi,\hat{H};\mathcal{A}) \label{eq:minimax_risk_def}\\
\intertext{and its associated error exponent (the limit below will be shown to exist)}
    \Gamma(\mathcal{P}_0,\mathcal{P}_1)&:= \lim_{S\rightarrow\infty} -\frac{1}{S} \log P_e^*(S). \label{eq:error_exp_def}
\end{align}

\subsection{Hellman-Cover~\cite{hellman1970learning}}\label{sec:hellman-cover}
Hellman and Cover \cite{hellman1970learning} studied the simple-versus-simple version\footnote{Hellman and Cover considered the asymptotic instantaneous error probability $\limsup_{n\rightarrow\infty}\mathbb{P}_{H=h}^{\pi}[\hat{H}(R_n)\ne h]$ rather than the asymptotic average error probability as we do here.
However, they use an aperiodic FSM and there is no adversary, so, due to ergodicity, both definitions coincide for them. Also, they study the Bayesian error probability rather than the worst-case one.}, i.e., $\mathcal{P}_0=\{p\},\mathcal{P}_1=\{q\}$ are singletons. As there is only one adversarial strategy, it is an inherently non-adversarial setting. They obtained the minimax risk as
\begin{align}
    P_e^*(S)&=\frac{1}{1+\sqrt{\gamma_\mathrm{HC}(p,q)^{S-1}}}, \label{eq:hellman_result}\\
\intertext{where}
    \gamma_\mathrm{HC}(p,q)&:= \max_{x_+\in\mathcal{X}} \frac{q(x_+)}{p(x_+)}\max_{x_-\in\mathcal{X}} \frac{p(x_-)}{q(x_-)}. \label{eq:gamma_HC_def}
\end{align}
For upper bounding $P_e^*(S)$, they use an FSM which they call a {\em saturable counter}:
fix $x_+\in\argmax_{x\in\mathcal{X}}\frac{q(x)}{p(x)}$ and $x_-\in \argmax_{x\in\mathcal{X}}\frac{p(x)}{q(x)}$.
From a state $s\in[S]\backslash\{1,S\}$, the FSM deterministically transitions to $s+1$ on observing $x_+$, which is treated as evidence favouring $H=1$; it transitions to $s-1$ on observing $x_-$, evidence favouring $H=0$; otherwise, it stays on at state $s$.
The behaviour is different at the ends:
from $s=1$ (resp. $s=S$), upon observing $x_-$ (resp. $x_+$), it stays put, as it cannot move further left (resp. right).
Additionally, on observing $x_+$ (resp. $x_-$), it transitions to state 2 (resp. $S-1$) with probability $\delta$ (resp. $\kappa\delta$), and with the remaining probability $1-\delta$ (resp. $1-\kappa\delta$) it stays put, where $\delta>0$ is a small positive probability and $\kappa>0$. In state 1, the FSM outputs $\hat{H}(1)=0$ and, in state $S$, $\hat{H}(S)=1$; outputs in the other states are fixed arbitrarily. 
The stationary distributions $\mu_0$ and $\mu_1$ of this FSM under $H=0$ and $H=1$, respectively, are such that
\begin{align}
    \frac{\mu_0(1)}{\mu_0(S)}= \kappa\left(\frac{p(x_-)}{p(x_+)}\right)^{S-1}&& 
    \frac{\mu_1(S)}{\mu_1(1)}= \frac{1}{\kappa}\left(\frac{q(x_+)}{q(x_-)}\right)^{S-1}\label{eq:hellman_achievability_1}
\end{align}
Reducing the probability of transitions out of the end states $\{1,S\}$ by choosing a $\delta\ll 1$, ensures that $\mu_h(1)+\mu_h(S)\approx 1$, $h=0,1$.
Then, the probabilities of the two types of errors are $\mu_0(S)$ and $\mu_1(1)$, which can now be obtained from \eqref{eq:hellman_achievability_1}.
Choosing $\kappa=\left(\frac{p(x_+)q(x_+)}{p(x_-)q(x_-)} \right)^{\frac{S-1}{2}}$ to make them equal leads to the achievability of \eqref{eq:hellman_result}.
The converse is shown by upper bounding $\max_{s,s'\in[S]}\frac{\mu_1(s)\mu_0(s')}{\mu_0(s)\mu_1(s')}$ of an arbitrary FSM by $\gamma^{S-1}$.

\subsection{Classical Robust Statistics/Composite Hypothesis Testing} \label{sec:composite}
In composite hypothesis testing, samples are generated i.i.d. according to $p$ or $q$, depending on the hypothesis, but $p$ and $q$ are unknown, except that $p\in\mathcal{P}_0$ and $q\in\mathcal{P}_1$ (where $\mathcal{P}_0$ and $\mathcal{P}_1$ need not be closed or convex). Classical robust statistics considers $\mathcal{P}_0$ and $\mathcal{P}_1$ which are ``neighbourhoods'' of a pair of known ideal model distributions $p_0$ and $p_1$, respectively~\cite{huber1965robust}. An important concept in this context is that of the least favourable distribution (LFD) pair.
\begin{defn}
    \label{def:LFD_pair}
    For $\mathcal{P}_0, \mathcal{P}_1\subseteq\Delta(\mathcal{X})$, $\left(p^*\in\mathcal{P}_0,q^*\in\mathcal{P}_1\right)$ is said to be an LFD pair, if for every $\eta>0$,
    \begin{align}
       \mathbb{P}_{X\sim p}\left[\frac{q^*(X)}{p^*(X)}>\eta\right] &\le \mathbb{P}_{X\sim p^*} \left[\frac{q^*(X)}{p^*(X)}>\eta\right],\; \forall p\in\mathcal{P}_0\label{eq:LFD_pair_def_0}\\
       \mathbb{P}_{X\sim q}\left[\frac{q^*(X)}{p^*(X)}>\eta\right] &\ge \mathbb{P}_{X\sim q^*}\left[\frac{q^*(X)}{p^*(X)}>\eta\right],\;\forall q\in\mathcal{P}_1, \label{eq:LFD_pair_def_1}
    \end{align}
    where if there is $x\in\mathcal{X}$ with $q^*(x)=p^*(x)=0$, we take $\frac{q^*(x)}{p^*(x)}=1$.
\end{defn}
It is known that an LFD pair exists when $\mathcal{P}_0$ and $\mathcal{P}_1$ are neighbourhoods of fixed distributions with respect to additive corruption, total-variation distance, Prohorov distance, etc.~\cite{huber1973minimax}.
Classically, when an LFD pair $(p^*,q^*)$ exists, the likelihood ratio test for $p^*$ vs $q^*$ performs at least as well as for testing $p$ versus $q$ for any pair $p\in\mathcal{P}_0, q\in\mathcal{P}_1$.
We observe that this continues to hold for memory constrained composite hypothesis testing.
When an LFD pair $(p^*,q^*)$ exists, Hellman and Cover's solution designed for testing $p^*$ versus $q^*$ performs at least as well for testing $p$ versus $q$ for any pair $p\in\mathcal{P}_0, q\in\mathcal{P}_1$.
In fact, we observe that this remains true for a weaker condition than the existence of an LFD pair.
\begin{defn}
    \label{def:weak_LFD_pair}
    For $\mathcal{P}_0, \mathcal{P}_1\subseteq\Delta(\mathcal{X})$, we call $\left(p^*\in\mathcal{P}_0,q^*\in\mathcal{P}_1\right)$ a {\em weak LFD pair} if the sets $\mathcal{X}_+:= \argmax_{x\in\mathcal{X}} \frac{q^*(x)}{p^*(x)}$, and $\mathcal{X}_-:= \argmin_{x\in\mathcal{X}}\frac{q^*(x)}{p^*(x)}$ satisfy
    \begin{align}
        p(\mathcal{X}_+)\le p^*(\mathcal{X}_+), &&p(\mathcal{X}_-)\ge p^*(\mathcal{X}_-),\; \forall p\in\mathcal{P}_0, \label{eq:weak_LFD_def_0}\\
        q(\mathcal{X}_+)\ge q^*(\mathcal{X}_+), &&q(\mathcal{X}_-)\le q^*(\mathcal{X}_-),\; \forall q\in\mathcal{P}_1, \label{eq:weak_LFD_def_1}
    \end{align}
    where if there is $x\in\mathcal{X}$ with $q^*(x)=p^*(x)=0$, we take $\frac{q^*(x)}{p^*(x)}=1$.
\end{defn}
\begin{rem}
    Note that an LFD pair is also a weak LFD pair, as can be seen by applying \eqref{eq:LFD_pair_def_0} and \eqref{eq:LFD_pair_def_1} to $\eta=\max_{x\in\mathcal{X}}\frac{q^*(x)}{p^*(x)}-\varepsilon$ (for small enough $\varepsilon>0$; if $\max_{x\in\mathcal{X}}\frac{q^*(x)}{p^*(x)}=\infty$, use a large enough finite $\eta$ instead), and  $\eta=\min_{x\in\mathcal{X}}\frac{q^*(x)}{p^*(x)}$. The converse is not true (see
    Section \ref{sec:example_weak_vs_strong}).
\end{rem}
For the {\em saturable counter} FSM, slightly modified to transition when any member of the sets $\mathcal{X}_+= \argmax_{x\in\mathcal{X}}\frac{q(x)}{p(x)}$ and $\mathcal{X}_-= \argmin_{x\in\mathcal{X}}\frac{q(x)}{p(x)}$ are observed, instead of only a fixed member each ($x_+$ and $x_-$, respectively), using  \eqref{eq:weak_LFD_def_0}-\eqref{eq:weak_LFD_def_1}, one can show that the asymptotic probability of error of distinguishing a pair $p\in\mathcal{P}_0$ vs $q\in\mathcal{P}_1$, is at most that of distinguishing $p^*$ vs $q^*$.\footnote{
To see this, suppose $p\in\mathcal{P}_0,q\in\mathcal{P}_1$ are the actual distributions of the samples. Then, \eqref{eq:weak_LFD_def_0} and \eqref{eq:weak_LFD_def_1} imply $\frac{p(\mathcal{X}_-)}{p(\mathcal{X}_+)}\ge \frac{p^*(\mathcal{X}_-)}{p^*(\mathcal{X}_+)}$ and $\frac{q(\mathcal{X}_+)}{q(\mathcal{X}_-)}\ge \frac{q^*(\mathcal{X}_+)}{q^*(\mathcal{X}_-)}$.
Together with \eqref{eq:hellman_achievability_1}, these imply $\frac{\mu_0(1)}{\mu_0(S)}\ge \kappa\left(\frac{p^*(\mathcal{X_-})}{p^*(\mathcal{X}_+)}\right)^{S-1}$ and $\frac{\mu_1(S)}{\mu_1(1)}\ge \frac{1}{\kappa} \left(\frac{q^*(\mathcal{X_+})}{q^*(\mathcal{X}_-)} \right)^{S-1}$.
Following the argument of~\cite{hellman1970learning}, as summarized in Section~\ref{sec:hellman-cover}, completes the proof.}

As we will show (see Theorem \ref{thm:weak_LFD}), this is true even in the adversarial setting, i.e., when a weak LFD pair $(p^*,q^*)$ exists, the (slightly modified) FSM above performs just as well against any adversarial strategy $\mathcal{A}$ as it does for testing $p^*$ versus $q^*$.

\section{Main Results}
\label{sec:results}
We define the following generalisations of $\gamma_\mathrm{HC}$.
\begin{defn}
    \label{def:gamma-and-C}
\begin{align}
    \gamma&:= \sup_{\substack{f_+,f_-:\mathcal{X}\rightarrow[0,1]\\ f_++f_-\le1}} \; \min_{\substack{p\in\mathcal{P}_0}, \, q\in\mathcal{P}_1} \frac{\mathbb{E}_q[f_+(X)]\mathbb{E}_p[f_-(X)]}{\mathbb{E}_q[f_-(X)]\mathbb{E}_p[f_+(X)]}, \label{eq:gamma_def}\\
    C&:= \sup_{f_+,f_-:\mathcal{X}\rightarrow[0,1]} \; \min_{\substack{p,p'\in\mathcal{P}_0,\;  q,q'\in\mathcal{P}_1}} \frac{\mathbb{E}_q[f_+(X)]\mathbb{E}_p[f_-(X)]}{\mathbb{E}_{q'}[f_-(X)]\mathbb{E}_{p'}[f_+(X)]}, \label{eq:C_def}
\end{align}
where, for these definitions, we follow the convention that $\frac{0}{0}=0$.
\end{defn}
\begin{rem}
    \label{rem:gamma_relation_to_pvp}
    Firstly, it is easy to see\footnote{Since
${\mathbb{E}_q[f_+]}/{\mathbb{E}_{p}[f_+]}
        = ({\sum_{x\in\mathcal{X}}q(x)f_+(x)})/({\sum_{x\in\mathcal{X}}p(x)f_+(x)})
        \le \max_{x\in\mathrm{supp}(f_+)}\frac{q(x)f_+(x)}{p(x)f_+(x)}$, we have 
$\sup_{f_+} \frac{\mathbb{E}_q[f_+]}{\mathbb{E}_{p}[f_+]}\le \max_{x\in\mathcal{X}}\frac{q(x)}{p(x)}$. 
Similarly,
    $\sup_{f_-} \frac{\mathbb{E}_p[f_-]}{\mathbb{E}_q[f_-]}= \max_{x\in\mathcal{X}}{p(x)}/{q(x)}$.}
that when $\mathcal{P}_0=\{p\}, \mathcal{P}_1=\{q\}$ are singletons, the supremum over $f_+$ and $f_-$ is $\gamma=\gamma_\mathrm{HC}(p,q)$ (and, clearly, $C=\gamma$).
    Now notice that swapping the min and sup in \eqref{eq:gamma_def} gives an upper bound on $\gamma$.
    \begin{align*}
        \gamma\le \min_{\substack{p\in\mathcal{P}_0}, \, q\in\mathcal{P}_1}  \sup_{\substack{f_+,f_-}} \frac{\mathbb{E}_q[f_+(X)]\mathbb{E}_p[f_-(X)]}{\mathbb{E}_q[f_-(X)]\mathbb{E}_p[f_+(X)]}.
    \end{align*}
    Since, for a fixed $p,q$, the inner supremum is just $\gamma_{\mathrm{HC}}(p,q)$,
    \begin{align}
    \gamma \leq \min_{\substack{p\in\mathcal{P}_0}, \, q\in\mathcal{P}_1} \gamma_{\mathrm{HC}}(p,q). \label{eq:gamma-le-gammaHC}
    \end{align}
    We show in Lemma \ref{lm:minimax} (page~\pageref{lm:minimax}) that $\sup$ and $\min$ in \eqref{eq:gamma_def} can be interchanged, i.e.,
    \begin{align}
         \gamma = \min_{\substack{p\in\mathcal{P}_0}, \, q\in\mathcal{P}_1} \gamma_{\mathrm{HC}}(p,q).
    \end{align}
\end{rem}

The functions $f_+$ and $f_-$ in~\eqref{eq:gamma_def} can be interpreted as follows: recall that, from state $s$, the FSM of Hellman and Cover transitions to $s+1$ on seeing $x_+$ and to $s-1$ on seeing $x_-$ (with modifications at the two end states as described earlier). The FSM we use to show our achievability result transitions from state $s$ to $s+1$ with {\em probability} $f_+(x)$ on seeing $x$ and to state $s-1$ with probability $f_-(x)$ (again with some modifications at the end states). When $f_+$ and $f_-$ are indicator functions of $x_+$ and $x_-$, respectively (or of sets $\mathcal{X}_+$ and $\mathcal{X}_-$, respectively), our FSM reduces to Hellman and Cover's (or the slightly modified one discussed in Section~\ref{sec:composite}). In general, restricting $f_+,f_-$ to indicator functions in \eqref{eq:gamma_def} may be strictly sub-optimal (see Section~\ref{sec:example_nonindicator}). 
Our main result is the following characterisation of $P_e^*(S)$.
\begin{thm}
    \label{thm:error_characterisation}
    For $S\ge2$,
    \begin{align}
        P_e^*(S)&\ge\frac{1}{1+\sqrt{\gamma^{S-1}}}. \label{eq:lower_bound}\\
   \intertext{and,}
        P_e^*(S)&\le\frac{1}{1+\sqrt{C\gamma^{S-2}}}. \label{eq:upper_bound}
    \end{align}
\end{thm}
\begin{rem}
    \label{rem:C_gamma_relation}
    Firstly, since the constant functions $f_+,f_-\equiv 1/2$ are feasible, we have $C,\gamma\ge1$.
    Secondly,
    \begin{align*}
        \min_{\substack{p,p'\in\mathcal{P}_0,\;  q,q'\in\mathcal{P}_1}} \frac{\mathbb{E}_q[f_+(X)]\mathbb{E}_p[f_-(X)]}{\mathbb{E}_{q'}[f_-(X)]\mathbb{E}_{p'}[f_+(X)]}
        &\le \min_{\substack{p,p'\in\mathcal{P}_0,\;  q,q'\in\mathcal{P}_1:\\ p=p',q=q'}} \frac{\mathbb{E}_q[f_+(X)]\mathbb{E}_p[f_-(X)]}{\mathbb{E}_{q'}[f_-(X)]\mathbb{E}_{p'}[f_+(X)]}\\
        &= \min_{\substack{p\in\mathcal{P}_0,\;  q\in\mathcal{P}_1}} \frac{\mathbb{E}_q[f_+(X)]\mathbb{E}_p[f_-(X)]}{\mathbb{E}_{q}[f_-(X)]\mathbb{E}_{p}[f_+(X)]}.
    \end{align*}
    Thus, we have $1\le C\le \gamma$.
    Hence, the lower and upper bounds \eqref{eq:lower_bound} and \eqref{eq:upper_bound} are within one state of each other: The minimax error that our $(S+1)$-state automaton achieves is no worse than the best possible error any $S$-state automaton can achieve.
    We show in Appendix \ref{sec:tight_lower_bound} that for $S=2$, \eqref{eq:upper_bound} is tight.
    See Section~\ref{sec:example_nonindicator} for an example for which $C<\gamma$.
\end{rem}
Since the upper and lower bounds agree to within one state, we have an exact characterisation of the error exponent.
\begin{coro}
    \label{coro:exponent_characterisation}
    \begin{align}
        \Gamma=\frac{1}{2}\log\gamma. \label{eq:exponent_characterisation}
    \end{align}
\end{coro}

When a weak LFD pair $(p^*,q^*)$ exists for $\mathcal{P}_0,\mathcal{P}_1$, $C=\gamma=\gamma_\mathrm{HC}(p^*,q^*)$ and hence Theorem~\ref{thm:error_characterisation} gives an exact characterisation of $P_e^*(S)$; see Appendix \ref{sec:results_extra} for a proof.
\begin{thm}
    \label{thm:weak_LFD}
    When a weak LFD pair $(p^*,q^*)$ exists for $\mathcal{P}_0,\mathcal{P}_1$,
    \begin{align}
        \label{eq:weak_LFD_achivability}
        C=\gamma=\gamma_\mathrm{HC}(p^*,q^*).
    \end{align}
\end{thm}

\section{Achievability}
\label{sec:achievability}
\begin{figure*}
    \centering
        \begin{tikzpicture}[scale=.8]
            \node[state] (1) at (3,0) {1};
            \node[state] (2) at (6,0) {2};
            \node[state,opacity=0] (3) at (9,0) {};
            \node[state] (s) at (12,0) {$s$};
            \node[state,opacity=0] (S-2) at (15,0) {};
            \node[state] (S-1) at (18,0) {$S-1$};
            \node[state] (S) at (21,0) {$S$};
                       
            \draw (1) edge[loop below] node {$1-\delta f_+^C(x)$} (1);
            \draw (1) edge[bend left] node {$\delta f_+^C(x)$} (2);
            \draw (2) edge[bend left] node (e1) {$f_-^\gamma(x)$ } (1);

            \draw (2) edge[loop above] node {$1-f_+^\gamma(x)-f_-^\gamma(x)$} (2);
            
            \draw (2) edge[bend left] node {$f_+^\gamma(x)$} (3);
            \draw (3) edge[bend left] node{$f_-^\gamma(x)$} (2);
        
            \draw (3) edge[bend left] node {$f_+^\gamma(x)$} (s);
            \draw (s) edge[bend left] node {$f_-^\gamma(x)$} (3);

            \draw (s) edge[loop above] node {$1-f_+^\gamma(x)-f_-^\gamma(x)$} (s);
            
            \draw (s) edge[bend left] node {$f_+^\gamma(x)$} (S-2);
            \draw (S-2) edge[bend left] node {$f_-^\gamma(x)$} (s);
                
            \draw (S-2) edge[bend left] node {$f_+^\gamma(x)$} (S-1);
            \draw (S-1) edge[bend left] node {$f_-^\gamma(x)$} (S-2);

            \draw (S-1) edge[loop above] node {$1-f_+^\gamma(x)-f_-^\gamma(x)$} (S-1);
            
            \draw (S-1) edge[bend left] node {$f_+^\gamma(x)$} (S);
            \draw (S) edge[bend left] node {$\kappa\delta f_-^C(x)$} (S-1);
            \draw (S) edge[loop below] node {$1-\kappa\delta f_-^C(x)$} (S);
        \end{tikzpicture}
        \caption{This figure shows the transition function $\pi(\cdot|\cdot, x)$. 
        Note that transitions out of the end states $1,S$ have a small probability (proportional to $\delta$).
        Also, when moving from the end states $1$ and $S$, the dependence on $x$ is determined by $f_+^C$ and $f_-^C$ respectively;
        whereas when moving from an internal state the dependence on $x$ is determined by $f_+^\gamma$ and $f_-^\gamma$ (depending on whether it is a rightwards move or a leftwards move).}
        \label{fig:pi}
\end{figure*}
To show our upper bound \eqref{eq:upper_bound}, taking inspiration from Hellman and Cover \cite{hellman1970learning}, we define an FSM very similar to theirs.
They identify letters $x_-,x_+\in\mathcal{X}$ which provide the strongest evidence for $H=0$ and $H=1$ respectively.
However, when the distribution of the samples under each hypothesis is not known precisely, it may be the case that no single letter provides strong evidence for $H=0$ (or $H=1$), but a linear combination does.
Thus, we generalise from letters $(x_+,x_-)$ to non-negative functions $(f_+,f_-)$ from $\mathcal{X}$ to $[0,1]$ (see the definitions of $\gamma$ and $C$ in \eqref{eq:gamma_def} and \eqref{eq:C_def}).

Given two pairs of such functions $(f^\gamma_+,f^\gamma_-)$, $(f_+^C,f_-^C)$ with $f^\gamma_++f^\gamma_-\le1$ and parameters $\delta,\kappa>0$, the FSM is shown in Figure~\ref{fig:pi}. From an internal state $s\in[S]\setminus\{1,S\}$, on observing $x$, the FSM transitions to state $s+1$ with probability $f^\gamma_+(x)$; to $s-1$ with probability $f^\gamma_-(x)$; and stays at $s$ with probability $1-f^\gamma_+(x)-f^\gamma_-(x)$. For the end state $s=1$ (resp. $s=S$), on observing $x$, it transitions to state $2$ (resp. $S-1$) with probability $\delta f_+^C(x)$  (resp. $\kappa\delta f_-^C(x)$) and stays on with the remaining probability. For example, when this FSM is in a state $s\in\{2,\ldots,S-1\}$, and the adversary has chosen $p$ to be the distribution of the next sample $X$, the FSM will, in the next step, move right with probability $\mathbb{E}_p[f^\gamma_+(X)]$ and move left with probability $\mathbb{E}_p[f^\gamma_-(X)]$.
For Hellman and Cover's FSM, these probabilities are $p(x_+)$ and $p(x_-)$ respectively. Let
\begin{align*}
    C_0:={}& \min_{p,p'\in\mathcal{P}_0} \frac{\mathbb{E}_p[f^C_-(X)]}{\mathbb{E}_{p'}[f^C_+(X)]}; &        \gamma_0:={}&\min_{p\in\mathcal{P}_0} \frac{\mathbb{E}_p[f^\gamma_-(X)]}{\mathbb{E}_{p}[f^\gamma_+(X)]};\\
    C_1:={}& \min_{q,q'\in\mathcal{P}_1} \frac{\mathbb{E}_q[f^C_+(X)]}{\mathbb{E}_{q'}[f^C_-(X)]}; &
    \gamma_1:={}&\min_{q,\in\mathcal{P}_1} \frac{\mathbb{E}_q[f^\gamma_+(X)]}{\mathbb{E}_{q}[f^\gamma_-(X)]};
\end{align*}
and $C(f^C_+,f^C_-)=C_1C_0$, $\gamma(f^\gamma_+,f^\gamma_-)=\gamma_1\gamma_0$.
Note that ${\mathbb{E}_p[f^\gamma_-(X)]}/{\mathbb{E}_{p}[f^\gamma_+(X)]}$ is the ratio of the probabilities of leftward transition to rightward transition at any internal state $s\in[S]\setminus\{1,S\}$ when the adversary chooses the distribution $p\in\mathcal{P}_0$; and $\gamma_0$ may be interpreted as the smallest value for this ratio that the adversary can induce under $H=0$. Similarly, $\kappa C_0$ is the smallest value of the ratio of probabilities of leftward transition from $S$ and rightward transition from state 1 that the adversary can induce under $H=0$. Notice that here the adversary may choose different distributions, $p$ in state $S$ and $p'$ in state 1, which explains the definition of $C_0$ (and the need for considering $(f_+^C,f_-^C)$ separate from $(f_+^\gamma,f_-^\gamma)$). For Hellman-Cover, $C_0=\gamma_0=p(x_-)/p(x_+)$.

The overall plan of the analysis is similar to Hellman and Cover's. We will show that (cf. \eqref{eq:hellman_achievability_1})
\begin{align}
    \liminf_{n\rightarrow\infty}\frac{\mu^0_n(1)}{\mu^0_n(S)}&\ge \kappa C_0\gamma_0^{S-2}, 
    \label{eq:correctness_to_error_ratio_0}\\
    \liminf_{n\rightarrow\infty}\frac{\mu^1_n(S)}{\mu^1_n(1)}&\ge\frac{C_1\gamma_1^{S-2}}{\kappa}, 
    \label{eq:correctness_to_error_ratio_1}
\end{align}
where, for $s\in[S],h\in\{0,1\}$,
\begin{align}
    \mu^h_n(s):=\frac{1}{n}\sum_{i\in[n]} \mathbb{P}_{H=h}^{\pi,\mathcal{A}}[R_{i-1}=s].
\end{align}
For a given $\varepsilon>0$, we will then argue that by choosing $\delta>0$ sufficiently small, we have 
\begin{align}
    \liminf_{n\rightarrow\infty} \left(\mu^h_n(1)+\mu^h_n(S)\right)\ge 1-\varepsilon, \label{eq:low_interior_occupancy}
\end{align}
for $h\in\{0,1\}$.
We can then conclude that
\begin{align}
    \limsup_{n\rightarrow\infty}1-\mu^0_n(1) \le{}& \frac{1+\varepsilon \kappa C_0\gamma_0^{S-2}}{1+\kappa C_0\gamma_0^{S-2}}, \text{ and} \label{eq:ach_error_0}\\
    \limsup_{n\rightarrow\infty}1-\mu^1_n(S) \le{}& \frac{1 + \varepsilon ({C_1\gamma_1^{S-2}}/{\kappa}) }{1+({C_1\gamma_1^{S-2}}/{\kappa})}. \label{eq:ach_error_1}
\end{align}
Set the decision function as $\hat{H}(1)=0$ and $\hat{H}(S)=1$ and choose $\kappa^2:={{C_1\gamma_1^{S-2}}/({C_0\gamma_0^{S-2}})}$.
Letting $\varepsilon\rightarrow0$,
\begin{align}
P_e^*(S)\le\frac{1}{1+\sqrt{C(f^C_+,f^C_-)\gamma(f^\gamma_+,f^\gamma_-)^{S-2}}}. \label{eq:ach_Pe}
\end{align}
We will actually show \eqref{eq:correctness_to_error_ratio_0}, \eqref{eq:correctness_to_error_ratio_1}, and \eqref{eq:low_interior_occupancy} only for $(f^C_+,f^C_-,f^\gamma_+,f^\gamma_-)$ satisfying, $\forall x\in\mathcal{X}$,
\begin{align}
    f_+^\gamma(x)\ge{}&\eta, & f_-^\gamma(x)\ge{}&\eta, \label{eq:f_gamma_regularisation}\\
    f_+^C(x)\ge{}&\eta, & f_-^C(x)\ge{}&\eta. \label{eq:f_C_regularisation}
\end{align}
for some $\eta>0$.
Letting $\varepsilon\rightarrow0$ and choosing $\kappa$ appropriately, we can still conclude \eqref{eq:ach_Pe} for such  $(f^C_+,f^C_-,f^\gamma_+,f^\gamma_-)$.
To obtain the result over the entire domain, first pick a sequence $\{ (f^\gamma_{+,m}, f^\gamma_{-,m}) \}_{m\in\mathbb{N}}$ approaching the supremum value $\gamma$.
Note that $\forall m\in\mathbb{N},\, \min_{p\in\mathcal{P}_0} \mathbb{E}_p[f^\gamma_{-,m}]>0$ and $\min_{q\in\mathcal{P}_1} \mathbb{E}_q[f^\gamma_{+,m}]>0$.
Define $ (f^\gamma_{+,m}{}', f^\gamma_{-,m}{}'):=(1-2\eta)(f^\gamma_{+,m}, f^\gamma_{-,m})+\eta$ (component-wise).
Then $ (f^\gamma_{+,m}{}', f^\gamma_{-,m}{}')$ satisfy \eqref{eq:f_gamma_regularisation}-\eqref{eq:f_C_regularisation} and $\lim_{\eta\rightarrow0} \gamma(f^\gamma_{+,m}{}',f^\gamma_{-,m}{}')= \gamma(f^\gamma_{+,m},f^\gamma_{-,m})$.
Now, taking $m\rightarrow\infty$, we get the result for $\gamma$.
A similar argument works for $C$.

It only remains to show \eqref{eq:correctness_to_error_ratio_0}, \eqref{eq:correctness_to_error_ratio_1}, and \eqref{eq:low_interior_occupancy} for $(f^C_+,f^C_-,f^\gamma_+,f^\gamma_-)$ satisfying \eqref{eq:f_gamma_regularisation}, \eqref{eq:f_C_regularisation}.
Define, for $h\in\{0,1\}$, $x\in\mathcal{X}$, and $s\in[S]$,  
\begin{align}
    \bar{p}^h_n(x|s)&:= \frac{\sum_{i\in[n]}\mathbb{E}_{H=h}^{\pi,\mathcal{A}}[\mathbf{1}_{R_{i-1}=s}P_i(x)]}{n\mu^h_n(s)},
\end{align}
where $P_i\in \mathcal{P}_h$ denotes the distribution chosen by the adversary for the $i$-th sample.
When $\mu^h_n(s)=0$, set $\bar{p}^h_n(\cdot|s)=p$ for an arbitrary $p\in\mathcal{P}_h$.
Although the adversary may choose $P_i\in \mathcal{P}_h$ based on the entire history $(X_1^{i-1},R_0^{i-1})$, in $\bar{p}^h_n$ only the effect of the current state $R_{i-1}$ is retained while the dependence on the other factors is averaged out.
It turns out that it is enough to study $\bar{p}^h_n$ for our purposes.
Let $w_i^h(s):=\mathbb{P}_{H=h}^{\pi,\mathcal{A}}[R_{i-1}=s]$, $i\in[n]$, $s\in[S]$.
Then, for $s\in[S]$ such that $\mu_n^h(s)>0$,
\begin{align*}
    \bar{p}^h_n(.|s)&= \sum_{\substack{i\in[n]\\w_i^h(s)>0}} 
    \frac{w_i^h(s)}{n\mu^h_n(s)} \mathbb{E}_{H=h}^{\pi,\mathcal{A}}[P_i(.)|R_{i-1}=s].
\end{align*}
Since $P_i\in\mathcal{P}_h$ and $\mathcal{P}_h$ is closed and convex, $\mathbb{E}_{H=h}^{\pi,\mathcal{A}}[P_i(.)|R_{i-1}=s]\in\mathcal{P}_h$. Since the coefficients of the conditional expectations in the sum above are non-negative and sum up to 1, $\bar{p}^h_n(.|s)\in\mathcal{P}_h$.
Let
\begin{align}
    T^h_n(t|s):=\sum_{x\in\mathcal{X}}\bar{p}^h_n(x|s) \pi(t|s,x),\quad s,t\in[S]. \label{eq:achievability-T}
\end{align}
Then,
\begin{align}
    (&T^h_n\mu^h_n)(t)\notag\\
    &:=\sum_{\substack{s\in[S]:\\ \mu^h_n(s)>0}} T^h_n(t|s) \mu^h_n(s)\notag\\
    &= \sum_{\substack{s\in[S]:\\ \mu^h_n(s)>0}} \sum_{x\in\mathcal{X}} \frac{\sum_{i\in[n]}\mathbb{E}_{H=h}^{\pi,\mathcal{A}}[\mathbf{1}_{R_{i-1}=s}P_i(x)]}{n\mu^h_n(s)} \mu^h_n(s) \pi(t|s,x)\notag\\
    &= \frac{1}{n}\sum_{i\in[n]} \mathbb{E}_{H=h}^{\pi,\mathcal{A}}\left[ \sum_{\substack{s\in[S]:\\ \mu^h_n(s)>0}} \sum_{x\in\mathcal{X}}\mathbf{1}_{R_{i-1}=s}P_i(x)\pi(t|s,x)\right]\notag\\
    &=\frac{1}{n}\sum_{i\in[n]} \mathbb{P}_{H=h}^{\pi,\mathcal{A}}[R_i=t], \label{eq:achievability-Tmu}
\end{align}
where the last step follows from letting $\mathcal{G}_i=\sigma(R_0^{i-1},X_1^{i-1},P_i)$ and observing
\begin{align*}
  \mathbb{P}_{H=h}^{\pi,\mathcal{A}}[R_i=t] &= \mathbb{E}_{H=h}^{\pi,\mathcal{A}}\left[ \mathbb{P}_{H=h}^{\pi,\mathcal{A}}[R_i=t|\mathcal{G}_i]  \right],\text{ and}\\
  \mathbb{P}_{H=h}^{\pi,\mathcal{A}}[R_i=t|\mathcal{G}_i] 
   &=\sum_{x\in\mathcal{X}} \mathbb{P}_{H=h}^{\pi,\mathcal{A}}[X_i=x|\mathcal{G}_i]\\
     &\qquad\qquad\mathbb{P}_{H=h}^{\pi,\mathcal{A}}[R_i=t|X_i=x,\mathcal{G}_i]\\
   &=\sum_{x\in\mathcal{X}} P_i(x) \pi(t|R_{i-1},x).
\end{align*} 
Therefore, from \eqref{eq:achievability-Tmu},
\begin{align}
    (T^h_n\mu^h_n-\mu^h_n)(s)=\frac{1}{n}\left( \mathbb{P}_{H=h}^{\pi,\mathcal{A}}[R_n=s] - \mathbb{P}_{H=h}^{\pi,\mathcal{A}}[R_0=s]\right).
    \label{eq:almost_stationary}
\end{align}
Consider the sequence $\{(\mu^h_n,\bar{p}^h_n) \}_{n\in\mathbb{N}}$, whose terms take values in the compact set $\Delta([S])\times (\mathcal{P}_h)^S$.
Take any convergent subsequence $\{ ( \mu^h_{n_k}, \bar{p}^h_{n_k}) \}_{k\in\mathbb{N}}$ with limit $(\mu^h,\bar{p}^h)$. Then, $\bar{p}^h(.|s)\in \mathcal{P}_h$, for every $s\in[S]$.
Passing to the limit in \eqref{eq:almost_stationary} (and using the fact that $\pi$ remains fixed in the definition \eqref{eq:achievability-T} of $T_n^h$), we have
\begin{align}
    T^h\mu^h=\mu^h, \label{eq:stationarity_in_lim}
\end{align}
where $T^h(t|s):=\sum_{x\in\mathcal{X}}\bar{p}^h(x|s) \pi(t|s,x)$.

Since $f_+^\gamma,f_-^\gamma,f_+^C,f_-^C$ satisfy \eqref{eq:f_gamma_regularisation}-\eqref{eq:f_C_regularisation}, $T^h$ is a transition function of an irreducible Markov chain, with $\mu^h$ as its stationary distribution.
Now,
\begin{align}
    \frac{\mu^0(1)}{\mu^0(S)}={}& \prod_{s\in[S-1]}\frac{\mu^0(s)}{\mu^0(s+1)}\nonumber\\
    ={}& \frac{\mu^0(1)}{\mu^0(2)}\frac{\mu^0(S-1)}{\mu^0(S)} \prod_{s\in[S-3]}\frac{\mu^0(s+1)}{\mu^0(s+2)}\nonumber\\
    ={}& \frac{\mathbb{E}_{\bar{p}^0(\cdot|2)}[f_-^\gamma]}{\mathbb{E}_{\bar{p}^0(\cdot|1)}[\delta f_+^C]} \frac{\mathbb{E}_{\bar{p}^0(\cdot|S)}[\kappa\delta f_-^C]}{\mathbb{E}_{\bar{p}^0(\cdot|S-1)}[f_+^\gamma]} \prod_{s\in[S-3]} \frac{\mathbb{E}_{\bar{p}^0(\cdot|s+2)} [f_-^\gamma]}{\mathbb{E}_{\bar{p}^0(\cdot|s+1)}[f_+^\gamma]} \nonumber\\
    ={}& \kappa\frac{\mathbb{E}_{\bar{p}^0(\cdot|S)}[ f_-^C]}{\mathbb{E}_{\bar{p}^0(\cdot|1)}[ f_+^C]} \prod_{s\in[S-2]} \frac{\mathbb{E}_{\bar{p}^0(\cdot|s+1)} [f_-^\gamma]}{\mathbb{E}_{\bar{p}^0(\cdot|s+1)}[f_+^\gamma]} \nonumber\\
    \ge{}&\kappa C_0(f_+^C,f_-^C) \gamma_0(f_+^\gamma,f_-^\gamma)^{S-2}. \label{eq:correctness_to_error_ratio_0_proven}
\end{align}
The above argument assumes $S\ge 3$.
If $S=2$,
\begin{align*}
    \frac{\mu^0(1)}{\mu^0(2)}={}& \frac{\mathbb{E}_{\bar{p}^0}(\cdot|2)[\kappa\delta f_-^C]}{\mathbb{E}_{\bar{p}^0}(\cdot|1)[\delta f_+^C]}\\
    \ge{}& \kappa C_0(f_+^C,f_-^C).
\end{align*}
Hence \eqref{eq:correctness_to_error_ratio_0_proven} holds for $S\ge 2$.
Similar analysis of $\mu^1$ gives
\begin{align}
    \frac{\mu^1(S)}{\mu^1(1)}\ge \frac{1}{\kappa} C_1(f_+^C,f_-^C) \gamma_1(f_+^\gamma,f_-^\gamma)^{S-2} \label{eq:correctness_to_error_ratio_1_proven}
\end{align}
for $S\ge2$.

Now we show \eqref{eq:low_interior_occupancy}.
Using \eqref{eq:f_gamma_regularisation} and \eqref{eq:f_C_regularisation} and applying detailed balance gives, for $s\in\{2,\ldots,S-1\}$,
\begin{align*}
    \frac{\mu^h(s)}{\mu^h(1)}={}& \prod_{i\in[s-1]} \frac{\mu^h(i+1)}{\mu^h(i)}\\
    ={}& \frac{\mu^h(2)}{\mu^h(1)}\prod_{i=2}^{s-1} \frac{\mu^h(i+1)}{\mu^h(i)}\\
    \le{}& \frac{\delta}{\eta}\left(\frac{1}{\eta}\right)^{s-2}\\
    ={}&\frac{\delta}{\eta^{s-1}}\\
    \le{}& \frac{\delta}{\eta^{S-2}}.
\end{align*}
Hence, by choosing $\delta$ small enough we can ensure $(\mu^h(1)+\mu^h(S))\ge1-\varepsilon$.
Since \eqref{eq:correctness_to_error_ratio_0_proven}, \eqref{eq:correctness_to_error_ratio_1_proven} and \eqref{eq:low_interior_occupancy} hold for every convergent subsequence of $\{\mu^h_n, \bar{p}^h_n\}_{n\in\mathbb{N}}$, the result holds for the liminf as well.

\section{Converse}
\label{sec:converse}
First note that there is an obvious lower bound, which is a corollary of Hellman and Cover's lower bound for the simple-vs-simple case. Suppose the adversary's strategy is to always choose $p\in\mathcal{P}_0$ under $H=0$ and $q\in\mathcal{P}_1$ under $H=1$. Then, for any FSM with $S$ states, the asymptotic average risk is at least that given by \eqref{eq:hellman_result}. By maximizing over $p$ and $q$, we have the following lower bound on the asymptotic minimax risk.
\begin{lm}[Corollary of Theorem 2 in \cite{hellman1970learning}]
    \label{lm:Hellman_converse}
    \begin{align}
        P_e^*(S)\ge\frac{1}{1+\bar{\gamma}^\frac{S-1}{2}},\label{eq:Hellman_converse}
    \end{align}
    where
    \begin{align}
        \bar{\gamma}:=\min_{\substack{p\in\mathcal{P}_0,\;q\in\mathcal{P}_1 }}\gamma_\mathrm{HC}(p,q). \label{eq:gamma_bar_def}
    \end{align}
\end{lm}

As pointed out in Remark \ref{rem:gamma_relation_to_pvp}, $\gamma$ and $\bar{\gamma}$ are $\sup$-$\min$ and $\min$-$\sup$ versions of the same optimisation problem. Our converse result \eqref{eq:lower_bound} follows from the fact that $\min$-$\sup$ equals $\sup$-$\min$ for this problem.
\begin{lm}
    \label{lm:minimax}
    \begin{align}
        \min_{p,q} \sup_{\substack{f_+,f_-}} \gamma(f_+,f_-;p,q)
        = \sup_{\substack{f_+,f_-}} \min_{p,q} \;\gamma(f_+,f_-;p,q), \label{eq:minimax_lemma}
    \end{align}
    where the minima are over $p\in\mathcal{P}_0$, $q\in\mathcal{P}_1$; the suprema are over non-negative functions $f_+$ and $f_-$ defined on $\mathcal{X}$ s.t. $f_+(x)+f_-(x)\leq 1$, $x\in\mathcal{X}$; and $\gamma(f_+,f_-;p,q)=\frac{\mathbb{E}_q[f_+(X)]\mathbb{E}_p[f_-(X)]}{\mathbb{E}_q[f_-(X)]\mathbb{E}_p[f_+(X)]}$.
\end{lm}
The above lemma would have followed from Sion's minimax theorem~\cite{sion1958general} had $\gamma(f_+,f_-;p,q)$ been, for a fixed $(f_+,f_-)$, quasi-convex (and lower semicontinuous) jointly in $(p,q)$, which has a compact domain; and quasi-concave (and upper semicontinuous) jointly in $(f_+,f_-)$ for a fixed $(p,q)$. However, this is not the case in general. Nevertheless, we are able to prove the lemma by rewriting the L.H.S. of \eqref{eq:minimax_lemma} as
\begin{align*}
     \min_{\substack{p\in\mathcal{P}_0}}\;\min_{\substack{q\in\mathcal{P}_1}}\;\sup_{f_+}\; \sup_{f_-} \;\gamma(f_+,f_-;p,q),
\end{align*}
where the suprema are over $f_+$ and $f_-$ which take values in\footnote{The equivalence can be seen by noting that $\gamma(f_+,f_-;p,q)$ is invariant under rescaling $f_+$ or $f_-$, separately. Hence, the optimal value remains the same whether we ask for $\sum_xf_+(x)=\sum_xf_-(x)=1$, or for $\max_xf_+(x)+f_-(x)\le1$.} $\Delta(\mathcal{X})= \{f:\mathcal{X}\rightarrow [0,1]:\sum_x f(x)=1\}$, and then applying Sion's minimax theorem four times to swap each $\min$ with each $\sup$.

\begin{proof}[Proof of Lemma \ref{lm:minimax}]
    The key ingredient is Sion's minimax theorem over extended reals \citep[Theorem 2.11]{ben2023new}.
    \begin{thm}
        \label{thm:sion_minimax}
        Let $\mathbf{X}$ be a convex subset of a topological vector space, and let $\mathbf{Y}$ be a convex, compact subset of a topological vector space.
        If $\gamma:\mathbf{X}\times\mathbf{Y} \rightarrow \mathbb{R}\cup\{\pm\infty\}$ is an extended real-valued function, satisfying
        \begin{itemize}
            \item $\gamma(\cdot,y)$ is upper semicontinuous and quasi-concave on $\mathbf{X}$ for every fixed $y\in\mathbf{Y}$, and
            \item $\gamma(x,\cdot)$ is lower semicontinuous and quasi-convex on $\mathbf{Y}$ for every fixed $x\in\mathbf{X}$,
        \end{itemize}
        then
        \begin{align}
            \sup_{x\in\mathbf{X}} \min_{y\in\mathbf{Y}} \gamma(x,y) = \min_{y\in\mathbf{Y}} \sup_{x\in\mathbf{X}} \gamma(x,y). \label{eq:sion_minimax}
        \end{align}
    \end{thm}
    The basic idea of the proof is that due to $\min_i\frac{a_i}{b_i}\le\frac{\sum_ia_i}{\sum_ib_i}\le \max_i\frac{a_i}{b_i}$,
    the function $\gamma(f,g;p,q)$ is quasi-convex and quasi-concave in each of its four inputs when the other three are fixed.
    However, $\gamma(f,g;p,q)$ is not jointly quasi-concave in $(f,g)$, or quasi-convex in $(p,q)$.
    Thus, to prove the lemma, we will write our optimisation problem as
    \begin{align*}
        \gamma=\sup_{\substack{f,g: \mathcal{X}\rightarrow[0,1]}} \min_{\substack{q\in\mathcal{P}_1 \\ p\in\mathcal{P}_0}} \gamma(f,g;p,q)={}& \sup_{f:\mathcal{X}\rightarrow[0,1]} \sup_{g:\mathcal{X}\rightarrow[0,1]} \min_{p\in\mathcal{P}_0}\min_{q\in\mathcal{P}_1}  \gamma(f,g;p,q),
    \end{align*}
    and then repeatedly apply Sion's theorem to swap each $\min$ with each $\sup$.

    Before we start applying Sion's theorem, we need to ensure its conditions are met.
    We set $\mathbf{Y}$ to either $\mathcal{P}_0$ or $\mathcal{P}_1$ as the case may be.
    Both of these are convex and compact.
    Note that $\gamma(f,g;p,q)$ is invariant to positive rescaling of $f$ or $g$. 
    Also, $f\equiv0$ or $g\equiv0$ are never optimisers due to our convention $\frac{0}{0}=0$.
    Thus, without loss of generality we may supremise $f,g$ over $\Delta(\mathcal{X})$.
    We need $\gamma(f,g;p,q)$ to be upper semicontinuous in $f$ and $g$ over its domain.
    The $\frac{0}{0}=0$ convention may cause this to fail.
    To fix this, we restrict the domain of $g$ and $f$ to $\Delta_0(\mathcal{X}):=\{g \in \Delta(\mathcal{X}):\min_{p\in\mathcal{P}_0}\mathbb{E}_p[g]>0\}$ and $\Delta_1(\mathcal{X}):=\{f\in\Delta(\mathcal{X}):\min_{q\in\mathcal{P}_1}\mathbb{E}_q[f]>0\}$, respectively.
    These sets are convex and dense in $\Delta(\mathcal X)$. Moreover, $\gamma(f,g;p,q)$ is jointly lower semicontinuous in $(f,g,p,q)$. Since $\mathcal P_0\times\mathcal P_1$ is compact, $\min_{p,q}\gamma(f,g;p,q)$ is lower semicontinuous in $(f,g)$. Consequently, restricting $f$ and $g$ to $\Delta_1$ and $\Delta_0$, respectively, does not change the supremum.
    We set $\mathbf{X}$ to be either $\Delta_0(\mathcal{X}):=\{g\in \Delta(\mathcal{X}):\min_{p\in\mathcal{P}_0}\mathbb{E}_p[g]>0\}$ or $\Delta_1(\mathcal{X}):=\{f\in\Delta(\mathcal{X}):\min_{q\in\mathcal{P}_1}\mathbb{E}_q[f]>0\}$ as the case may be.
    With this choice of domain, $\gamma(f,g;p,q)$ is upper semicontinuous in $f$ and $g$.
    It is easy to check that $\Delta_0(\mathcal{X})$ and $\Delta_1(\mathcal{X})$ are convex.
    Therefore, we have
    \begin{align*}
        \gamma={}&  \sup_{f\in\Delta_1(\mathcal{X})} \sup_{g\in\Delta_0(\mathcal{X})} \min_{p\in\mathcal{P}_0} \min_{q\in\mathcal{P}_1} \gamma(f,g;p,q).
    \end{align*}
    We need one more fact for this proof, which is that the function $\gamma(f,g;p,q)$ splits as follows:
    \begin{align*}
        \gamma(f,g;p,q)={}& \frac{\mathbb{E}_q[f]\mathbb{E}_p[g]}{\mathbb{E}_q[g]\mathbb{E}_p[f]},\\
        ={}& \frac{\mathbb{E}_q[f]}{\mathbb{E}_p[f]} \cdot \frac{\mathbb{E}_p[g]}{\mathbb{E}_q[g]},\\
        =:{}&\gamma_+(f;p,q)\gamma_-(g;p,q),
    \end{align*}
    and also as
    \begin{align*}
        \gamma(f,g;p,q)={}& \frac{\mathbb{E}_q[f]\mathbb{E}_p[g]}{\mathbb{E}_q[g]\mathbb{E}_p[f]},\\
        ={}& \frac{\mathbb{E}_q[f]}{\mathbb{E}_q[g]} \cdot \frac{\mathbb{E}_p[g]}{\mathbb{E}_p[f]},\\
        =:{}& \gamma_1(f,g;q)\gamma_0(f,g;p).
    \end{align*}
    Note that the functions $\gamma,\gamma_+,\gamma_-, \gamma_0,$ and $\gamma_1$ all have convex domain for $f$ and/or $g$, and convex, compact domain for $p$ and/or $q$.
    Furthermore, $\gamma_+$ and $\gamma_-$ are quasi-concave in $f$ and $g$ respectively; and $\gamma_0$ and $\gamma_1$ are quasi-convex in $p$ and $q$ respectively.

    \paragraph{Step 1} We want to apply Theorem \ref{thm:sion_minimax} to $\min_{q\in\mathcal{P}_1} \gamma(f,g;p,q)$ with arguments $g$ and $p$.
    Firstly, note that it is quasi-concave and upper semicontinuous in $g$, since $\gamma(f,g;p,q)$ is quasi-concave and upper semicontinuous in $g$ and minimisation preserves quasi-concavity and upper semicontinuity. 
    Note that
    \begin{align*}
        \min_{q\in\mathcal{P}_1} \gamma(f,g;p,q)={}& \gamma_0(f,g;p) \min_{q\in\mathcal{P}_1} \gamma_1(f,g;q).
    \end{align*}
    Since, $\min_{q\in\mathcal{P}_1} \gamma_1(f,g;q)$ is independent of $p$ and $\gamma_0(f,g;p)$ is quasi-convex and lower semicontinuous in $p$, we have that $\min_{q\in\mathcal{P}_1} \gamma(f,g;p,q)$ is also quasi-convex and lower semicontinuous in $p$.
    Hence, $\min_{q\in\mathcal{P}_1} \gamma(f,g;p,q)$ with arguments $g$ and $p$ satisfies the requirements of Theorem \ref{thm:sion_minimax}.
    Therefore,
    \begin{align}
        \sup_{g\in\Delta_0(\mathcal{X})} \min_{p\in\mathcal{P}_0} \min_{q\in\mathcal{P}_1} \gamma(f,g;p,q) =  \min_{p\in\mathcal{P}_0} \sup_{g\in\Delta_0(\mathcal{X})} \min_{q\in\mathcal{P}_1} \gamma(f,g;p,q).\label{eq:step1}
    \end{align}
    \paragraph{Step 2}
    We now want to apply Theorem \ref{thm:sion_minimax} to $\sup_{g\in\Delta_0(\mathcal{X})} \min_{q\in\mathcal{P}_1} \gamma(f,g;p,q) = \min_{q\in\mathcal{P}_1} \sup_{g\in\Delta_0(\mathcal{X})}  \gamma(f,g;p,q)$.
    (This exchange is valid since $\gamma(f,g;p,q)$ with arguments $g$ and $q$ satisfies the requirements of Theorem \ref{thm:sion_minimax}.)
    
    Since $\min_{q\in\mathcal{P}_1}\gamma(f,g;p,q)$ is quasi-convex and lower semicontinuous in $p$ and $\sup$ preserves quasi-convexity and lower semicontinuity, $\sup_{g\in\Delta_0(\mathcal{X})} \min_{q\in\mathcal{P}_1}\gamma(f,g;p,q)$ is also quasi-convex and lower semicontinuous in $p$.
    We will show in Step 3 that $\sup_{g\in\Delta_0(\mathcal{X})} \gamma(f,g;p,q)$ is quasi-concave and upper semicontinuous in $f$.
    Then, since minimisation preserves quasi-concavity and upper semicontinuity, $\min_{q\in\mathcal{P}_1} \sup_{g\in\Delta_0(\mathcal{X})} \gamma(f,g;p,q)= \sup_{g\in\Delta_0(\mathcal{X})} \min_{q\in\mathcal{P}_1} \gamma(f,g;p,q)$ is also quasi-concave and upper semicontinuous in $f$.
    Therefore,
    \begin{align}
         \sup_{f\in\Delta_1(\mathcal{X})}  \min_{p\in\mathcal{P}_0} \sup_{g\in\Delta_0(\mathcal{X})} \min_{q\in\mathcal{P}_1} \gamma(f,g;p,q) ={}& \min_{p\in\mathcal{P}_0} \sup_{f\in\Delta_1(\mathcal{X})} \sup_{g\in\Delta_0(\mathcal{X})} \min_{q\in\mathcal{P}_1} \gamma(f,g;p,q)\nonumber\\
         ={}& \min_{p\in\mathcal{P}_0} \sup_{f\in\Delta_1(\mathcal{X})}  \min_{q\in\mathcal{P}_1} \sup_{g\in\Delta_0(\mathcal{X})} \gamma(f,g;p,q). \label{eq:step2}
    \end{align}

    \paragraph{Step 3}
    For the final step we want to apply Theorem \ref{thm:sion_minimax} to $\sup_{g\in\Delta_0(\mathcal{X})} \gamma(f,g;p,q)$.
    Note that $\sup_{g\in\Delta_0(\mathcal{X})} \gamma(f,g;p,q)$ is quasi-convex and lower semicontinuous in $q$ since $\gamma(f,g;p,q)$ is quasi-convex and lower semicontinuous in $q$ and $\sup$ preserves quasi-convexity and lower semicontinuity.
    
    As promised, we now check that $\sup_{g\in\Delta_0(\mathcal{X})} \gamma(f,g;p,q)$ is quasi-concave and upper semicontinuous in $f$.
    Note that
    \begin{align*}
        \sup_{g\in\Delta_0(\mathcal{X})} \gamma(f,g;p,q) ={}& \gamma_+(f;p,q) \sup_{g\in\Delta_0(\mathcal{X})} \gamma_-(g;p,q).
    \end{align*}
    Since $\sup_{g\in\Delta_0(\mathcal{X})} \gamma_-(g;p,q)$ is independent of $f$ and $\gamma_+(f;p,q)$ is quasi-concave and upper semicontinuous in $f$, $\sup_{g\in\Delta_0(\mathcal{X})} \gamma(f,g;p,q)$ is also quasi-concave and upper semicontinuous in $f$.
    Therefore,
    \begin{align}
        \sup_{f\in\Delta_1(\mathcal{X})}  \min_{q\in\mathcal{P}_1} \sup_{g\in\Delta_0(\mathcal{X})} \gamma(f,g;p,q) = \min_{q\in\mathcal{P}_1} \sup_{f\in\Delta_1(\mathcal{X})} \sup_{g\in\Delta_0(\mathcal{X})} \gamma(f,g;p,q). \label{eq:step3} 
    \end{align}
    Combining \eqref{eq:step1}, \eqref{eq:step2}, and \eqref{eq:step3} gives 
    \begin{align}
        \gamma = \min_{p\in\mathcal{P}_0} \min_{q\in\mathcal{P}_1} \sup_{f\in\Delta_1} \sup_{g\in \Delta_0} \gamma(f,g;p,q).
    \end{align}
    We can again replace the domains $\Delta_0(\mathcal{X})$ and $\Delta_1(\mathcal{X})$ with $\Delta(\mathcal{X})$ since $\Delta_0(\mathcal{X})$ and $\Delta_1(\mathcal{X})$ are dense in $\Delta(\mathcal{X})$ and $\gamma(f,g;p,q)$ is lower semicontinuous in $f,g$ for any fixed $p,q$; and then replace $\Delta(\mathcal{X})$ with $\{ f:\mathcal{X}\rightarrow[0,1]\}$ due to the scaling invariance of $\gamma(f,g;p,q)$ and since $f\equiv0$ or $g\equiv0$ are not optimisers.
\end{proof}
\section{Examples}
\label{sec:examples}
\subsection{Optimal $f_+$ may not be an indicator function}
\label{sec:example_nonindicator}
In this section, we give an example to illustrate the benefit of generalising Hellman and Cover's FSM to ours.
Let $\mathcal{X}=\{a,b,c\}$, $\mathcal{P}_0=\{p\in\Delta(\mathcal{X}):p(a)+2p(b)=\frac{1}{2}\}$, and $\mathcal{P}_1= \left\{ \left( \frac{1}{3},\frac{1}{3},\frac{1}{3}\right)\right\}$.
As we show in Appendix \ref{sec:example_nonindicator_extra},
for this example, the solution to \eqref{eq:gamma_def} is $\gamma=4$ and it has a unique maximiser (up to independent positive scaling) $f_+=\left(\frac{1}{3},1,0\right),f_-=(0,0,1)$.
Hence, from states $\{2,\ldots,S-1\}$, our FSM moves to the right with probability 1/3 on observing $a$, and deterministically when observing $b$; and it moves to the left on observing $c$.
This gives a strictly better exponent than picking some subsets $\mathcal{X}_+, \mathcal{X}_- \subseteq\mathcal{X}$, always moving right on observing $x\in\mathcal{X}_+$, and always moving left on observing $x\in\mathcal{X}_-$.

We also note that, for this example $C=3$ which is strictly less than $\gamma$ (see Appendix \ref{sec:example_C_less_than_gamma_extra} for a proof).
Thus, there is a gap between the lower and upper bounds, \eqref{eq:lower_bound} and \eqref{eq:upper_bound}, of Theorem \ref{thm:error_characterisation}.

\subsection{Weak LFD pairs may exist even if strong LFD pairs do not}
\label{sec:example_weak_vs_strong}
Consider the example $\mathcal{X}=\{a,b,c,d,e\}$, $\mathcal{P}_1=\{q\}=\left\{\left(\frac{9}{50}, \frac{4}{15}, \frac{4}{15}, \frac{4}{15}, \frac{1}{50} \right)\right\}$,
$\mathcal{P}_0=\{p_\lambda:\lambda\in[0,1]\}$, where $p_\lambda=\left(\frac{1}{50}, \frac{2\lambda}{5},\frac{1-\lambda}{5},\frac{3-\lambda}{5},\frac{9}{50} \right),\;\lambda\in[0,1]$.
We show in Appendix \ref{sec:example_weak_vs_strong_extra} that there exist weak LFD pairs for this example (specifically, $(p_\lambda,q)$ is a weak LFD pair for every $\lambda\in\left(\frac{2}{27},\frac{23}{27}\right)$), but no strong LFD pair exists.

\section{Deterministic FSM}
\label{sec:deterministic}
In this section we discuss the performance of time-invariant deterministic FSMs for the adversarial hypothesis testing problem.
A deterministic time-invariant, finite-state machine (DFSM) with state space [S], input space $\mathcal{X}$, and output space $\{0,1\}$, is specified by a tuple $(T, \hat{H}, R_0)$.
Here $T:[S]\times\mathcal{X}\rightarrow[S]$ is a deterministic state transition function, $\hat{H}:[S]\rightarrow\{0,1\}$ is a decision function, and $R_0\in[S]$ is the initial state.
Let $R_n,n\in\mathbb{N}$ denote the state of the DFSM at the end of step $n$.
During time step $n$, the detector receives a sample $X_n\in\mathcal{X}$ and moves to the state $T(R_{n-1}, X_n)$.
The detector's output at time $n$ is $\hat{H}(R_n)$.
Under hypothesis $H=h$, $h\in\{0,1\}$, the sample $X_n$ is drawn from a distribution $p_n$ chosen adversarially from $\mathcal{P}_h$.
When choosing $p_n$, the adversary, who knows the hypothesis, has access to past samples $X_1^{n-1}$ and the past sequence of detector states $R_0^{n-1}$.
Thus, an adversarial strategy $\mathcal{A}_h$, under hypothesis $H=h$, consists of a sequence of functions $\{p_n:{\mathcal{X}}^{n-1}\times{[S]}^{n}\rightarrow\mathcal{P}_h\}_{n\in\mathbb{N}\backslash\{0\}}$.

The performance of a DFSM against a pair of adversarial strategies $\mathcal{A}=(\mathcal{A}_0,\mathcal{A}_1)$ is measured by
\begin{align}
    P_{e,\mathrm{det}}(T,\hat{H};\mathcal{A}):=\max_{h\in\{0,1\}}\limsup_{n\rightarrow\infty}\frac{1}{n}\sum_{i\in[n]}\mathbb{P}_{H=h}^{T,\mathcal{A}_h}[\hat{H}(R_i)\ne h].
\end{align}
where we have suppressed the dependence on the starting state $R_0$.
The asymptotic minimax risk is defined as
\begin{align}
    P_{e,\mathrm{det}}^*(S):=\inf_{T,\hat{H}}\sup_\mathcal{A} P_{e,\mathrm{det}}(T,\hat{H};\mathcal{A}),
\end{align}
where the $\inf$ is over all $S$-state DFSMs, and the $\sup$ is over all adversarial strategies.
The error exponent is defined as
\begin{align}
    \Gamma_\mathrm{det}(\mathcal{P}_0,\mathcal{P}_1):=\liminf_{S\rightarrow\infty}-\frac{1}{S}\log P_{e,\mathrm{det}}^*(S).
\end{align}

The deterministic version of the problem is much more challenging than the randomised version.
Even in the non-adversarial case (i.e. when $\mathcal{P}_0= \{P\}, \mathcal{P}_1=\{Q\}$ are singletons), a tight characterisation of the error exponent is not known.
The best known results are in the Bernoulli case, i.e., $\mathcal{X}=\{0,1\}$. \cite{berg2020binary, shubert1973testing} consider the problem of testing $\mathrm{Ber}(p)$ against $\mathrm{Ber}(q)$ for $0<p<q<1$.
They obtain the following results:
\begin{thm}
    \textup{(\citep[Theorem 2]{berg2020binary})}
    \label{thm:non_adv_bin_det_achievability}
    When $\mathcal{X}=\{0,1\}$, for $0< p<q<1$,
    \begin{align}
        \Gamma_\mathrm{det}(\{\mathrm{Ber}(p)\},\{ \mathrm{Ber}(q)\}) \ge \frac{\log q\log(1-p)-\log p\log(1-q)}{\log p(1-p)+\log q(1-q)}.
        \label{eq:non_adv_bin_det_achievability}
    \end{align}
\end{thm}
\begin{thm}
    \textup{(\cite{shubert1973testing})}
    \label{thm:non_adv_sym_bin_det_achievability}
    When $\mathcal{X}=\{0,1\}$, for $0<p<1/2,q=1-p$, 
    \begin{align}
        \Gamma_\mathrm{det}(\{\mathrm{Ber}(p)\},\{ \mathrm{Ber}(q)\}) \ge \frac{1}{2}\log\frac{q}{p}\left(1+\frac{\log q}{\log p}\right). \label{eq:non_adv_sym_bin_det_achievability}
    \end{align}
\end{thm}
\begin{thm}
    \textup{(\citep[Theorem 1]{berg2020binary})}
    \label{thm:non_adv_bin_det_converse}
    \begin{align}
        \limsup_{S\rightarrow\infty}-\frac{1}{S}\log P_{e,\mathrm{det}}^*(S)\le -\frac{\log\min\{p,1-p\} \log\min\{q,1-q\}}{\log\min\{p,1-p\} + \log\min\{q,1-q\}} \label{eq:non_adv_bin_det_converse}
    \end{align}
\end{thm}
As discussed in \cite{berg2020binary}, Theorem \ref{thm:non_adv_bin_det_converse} demonstrates that there is an exponential gap between the capabilities of randomised and deterministic FSMs.

In the adversarial Bernoulli setting, we consider $\mathcal{P}_0,\mathcal{P}_1$ being sub-intervals of $[0,1]$.
Without loss of generality we can consider $\mathcal{P}_0= \{ \mathrm{Ber}(p'):0\le p'\le p\}$ and $\mathcal{P}_1=\{ \mathrm{Ber}(q'):q\le q'\le 1\}$, $0<p<q<1$.
The lower bound \eqref{eq:non_adv_bin_det_converse} carries over directly, since drawing samples from $\mathrm{Ber}(p),\mathrm{Ber}(q)$ under $H=0,1$ respectively is an admissible adversarial strategy. 
We now show that the upper bounds \eqref{eq:non_adv_bin_det_achievability}-\eqref{eq:non_adv_sym_bin_det_achievability} continue to hold in the adversarial setting.
\begin{thm}
    \label{thm:adv_bin_det_achievability}
    When $\mathcal{X}=\{0,1\}$, for $0< p<q<1$, $\mathcal{P}_0 =\{ \mathrm{Ber}(p'):0\le p'\le p\}, \mathcal{P}_1= \{\mathrm{Ber}(q'):q\le q'\le 1\}$,
    \begin{align}
        \Gamma_\mathrm{det}(\mathcal{P}_0,\mathcal{P}_1) \ge \frac{\log q\log(1-p)-\log p\log(1-q)}{\log p(1-p)+\log q(1-q)}.
    \end{align}
    Furthermore, if $q=1-p$,
     \begin{align}
        \Gamma_\mathrm{det}(\mathcal{P}_0,\mathcal{P}_1) \ge \frac{1}{2}\log\frac{q}{p}\left(1+\frac{\log q}{\log p}\right).
    \end{align}
\end{thm}
\begin{proof}
    We extend the proofs of Theorems \ref{thm:non_adv_bin_det_achievability} and \ref{thm:non_adv_sym_bin_det_achievability} through a coupling argument.
    Suppose $\{X_n\}_{n\in\mathbb{N}}$ are generated i.i.d according to $\mathrm{Ber}(p)$ or $\mathrm{Ber}(q)$ under $H=0$ or $H=1$, respectively.
    Suppose $\{Y_n\}_{n\in\mathbb{N}}$ are generated following some fixed adversarial strategy $\mathcal{A}_0$ or $\mathcal{A}_1$ under $H=0$ or $H=1$, respectively.
    Let $R_n^X$ (resp. $R_n^Y$) denote the state of the automaton after $n$ steps when it is fed $\{X_n\}_{n\in\mathbb{N}}$ (resp. $\{Y_n\}_{n\in\mathbb{N}}$).
    Then, there exists a coupling of these two distributions such that for all $n\in\mathbb{N}$,
    \begin{align}
        Y_n\le X_n \text{ under } H=0 \text{ and } Y_n\ge X_n \text{ under } H=1. \label{eq:coupling}
    \end{align}
    Specifically, let $U_n\sim\mathrm{Uniform}([0,1])$ i.i.d.
    Under $H=0$, define $\theta_n$ such that $\mathrm{Ber}(\theta_n)=p_n(Y_1^{n-1},(R^Y)_0^{n-1}), Y_n=\mathbf{1}_{U_n\le\theta_n}$ and $X_n=\mathbf{1}_{U_n\le p}$.
    Then $Y_n\le X_n$ follows from $\theta_n\le p$.
    A similar construction works under $H=1$.
    The automata used by \cite{berg2020binary} and \cite{shubert1973testing} to show Theorems \ref{thm:non_adv_bin_det_achievability} and \ref{thm:non_adv_sym_bin_det_achievability} respectively have states labelled with integers in the range $[-a,b]$ for $a,b\in\mathbb{N}\setminus\{0\}$. The two end states $-a$ and $b$ are absorbing states; all other states are transient.
    The initial state $R_0$ is 0.
    The decision function is $\hat{H}(s)= \mathbf{1}_{s\ge0}$.
    The automata also have the following monotonicity property:
    \begin{align}
        \forall s,s'\in\{-a,\ldots,b\},\ x,y\in\{0,1\},&\nonumber\\
        x\le y \text{ and } s\le s'\implies{}& T(s,x)\le T(s',y). \label{eq:automata_regularity}
    \end{align}
    Then, \eqref{eq:coupling} and \eqref{eq:automata_regularity} imply that, for all $n\ge0$
    \begin{align}
        R_n^Y\le{}& R_n^X \text{ under } H=0,\\
        R_n^Y\ge{}& R_n^X \text{ under } H=1.
    \end{align}
    Since $\hat{H}(s)=\mathbf{1}_{s\ge0}$ is non-decreasing, 
    \begin{align}
        \mathbb{P}_{H=0}^{\mathcal{A}_0}[\hat{H}(R^Y_n)=1] \le{}&  \mathbb{P}_{H=0}^{\mathcal{A}_0}[\hat{H}(R^X_n)=1],\\
        \text{and, } \mathbb{P}_{H=1}^{\mathcal{A}_1}[\hat{H}(R^Y_n)=1] \ge{}& \mathbb{P}_{H=1}^{\mathcal{A}_1}[\hat{H}(R^X_n)=1].
    \end{align}
    We conclude that
    \begin{align*}
        P_{e,\mathrm{det}}^{\mathcal{P}_0,\mathcal{P}_1}(T,\hat{H},\mathcal{A})\le P_{e,\mathrm{det}}^{\{\mathrm{Ber}(p)\},\{\mathrm{Ber}(q)\}}(T,\hat{H}).
    \end{align*}
    Since this was shown for an arbitrary adversarial strategy $\mathcal{A}=(\mathcal{A}_0,\mathcal{A}_1)$, the result follows.

    It only remains to show that the automata used for proving Theorems~\ref{thm:non_adv_bin_det_achievability} and \ref{thm:non_adv_sym_bin_det_achievability} in \cite{berg2020binary} and \cite{shubert1973testing}, respectively, satisfy \eqref{eq:automata_regularity}.
    We briefly describe these automata below.
    
    The automaton used by \cite{berg2020binary} is simple. It waits for a long run of consecutive zeros or ones and declares $H=0$ or $H=1$ depending on which run happens first.
    Specifically, if  a run of $a$ consecutive zeros are seen before a run of $b$ consecutive ones, the automaton gets trapped in state $-a$.
    Otherwise, if it sees a run of $b$ consecutive ones before a run of $a$ consecutive zeros, it gets trapped in state $b$.
    See Figure \ref{fig:automata_Berg}
    The initial state is $R_0=0$.
    $a$ and $b$ are chosen based on $p$ and $q$ as
    \begin{align*}
        a={}& \mathrm{round}\left(\frac{\log pq}{\log pq(1-p)(1-q)}S + \log\frac{\frac{(1-q)^2}{q}\log q(1-q)}{\frac{(1-p)^2}{p}\log p(1-p)}-1\right),\\
        b={}& S-a-1,
    \end{align*}
    where $\mathrm{round}$ rounds the value to the nearest integer.
    \begin{figure}
        \centering
        \begin{tikzpicture}[scale=0.7, every node/.style={scale=0.7}]
            \node[state] (1) at (3,0) {$-a$};
            \node[state] (2) at (6,0) {$-a+1$};
            \node[state,opacity=0] (3) at (9,0) {};
            \node[state] (s-1) at (12,0) {$-1$};
            \node[state,initial] (s) at (15,0) {$0$};
            \node[state] (s+1) at (18,0) {$+1$};
            \node[state,opacity=0] (S-2) at (21,0) {};
            \node[state] (S-1) at (24,0) {$b-1$};
            \node[state] (S) at (27,0) {$b$};

            \draw (s) edge[bend left] node[below]{1} (s+1);
            \draw (s+1) edge[bend left] node{1} (S-2);
            \draw (S-1) edge[bend left] node{1} (S);

            \draw (s+1) edge[bend left] node{0} (s-1);
            \draw (S-1) edge[bend left] node{0} (s-1);
            \draw (S) edge[loop above] node{0,1} (S);

            \draw (s) edge[bend left] node[above]{0} (s-1);
            \draw (s-1) edge[bend left] node{0} (3);
            \draw (2) edge[bend left] node{0} (1);

            \draw (s-1) edge[bend left] node{1} (s+1);
            \draw (2) edge[bend left] node[below]{1} (s+1);
            \draw (1) edge[loop below] node{0,1} (1);
        \end{tikzpicture}
        \caption{Automaton used by \cite{berg2020binary} to show Theorem \ref{thm:non_adv_bin_det_achievability}.}
        \label{fig:automata_Berg}
    \end{figure}
    It can easily be seen that this automaton satisfies \eqref{eq:automata_regularity}.

    The automaton used by \cite{shubert1973testing} is slightly more complex.
    The state space is $\{-(c+1),-c,\ldots,c,c+1\}$, with the number of states being $S=2c+3$.
    The transition function is
    \begin{align*}
        T(s,1)= s+1 &\hspace{3cm} \text{if } 0\le s\le c\\
        T(s,0)= s-1 &\hspace{3cm} \text{if } -c\le s\le 0\\
        T(s,0)=\min\{s-1,k\} &\hspace{3cm} \text{if } 1\le s \le c\\
        T(s,1)=\max\{s+1,-k\} &\hspace{3cm} \text{if } -c\le s\le -1\\
        T(\pm(c+1),0)=T(\pm(c+1),1)=\pm(c+1)&
    \end{align*}
    for $k=c\frac{\log q}{\log p}+\mathcal{O}(1)$ (see Figure \ref{fig:automata_Shubert}).
    \begin{figure}
        \centering
        \begin{tikzpicture}[scale=0.7, every node/.style={scale=0.7}]
            \node[state,initial] (0) at (0,0) {0};

            \node[state] (1) at (2,1) {1};
            \node[state,opacity=0] (2) at (4,1) {};
            \node[state] (k-1) at (6,1) {$k-1$};
            \node[state] (k) at (8,1) {$k$};
            \node[state] (k+1) at (10,1) {$k+1$};
            \node[state,opacity=0] (c-1) at (12,1) {};
            \node[state] (c) at (14,1) {$c$};
            \node[state] (c+1) at (16,1) {$c+1$};

            \node[state] (-1) at (2,-1) {$-1$};
            \node[state,opacity=0] (-2) at (4,-1) {};
            \node[state] (-k+1) at (6,-1) {$-(k-1)$};
            \node[state] (-k) at (8,-1) {$-k$};
            \node[state] (-k-1) at (10,-1) {$-(k+1)$};
            \node[state,opacity=0] (-c+1) at (12,-1) {};
            \node[state] (-c) at (14,-1) {$-c$};
            \node[state] (-c-1) at (16,-1) {$-(c+1)$};

            \draw (0) edge node[below]{1} (1);
            \draw (1) edge[bend right] node{1} (2);
            \draw (k-1) edge[bend right] node{1} (k);
            \draw (k) edge[bend right] node{1} (k+1);
            \draw (k+1) edge[bend right] node{1} (c-1);
            \draw (c) edge[bend right] node{1} (c+1);
            \draw (c+1) edge[loop above] node{0,1} (c+1);

            \draw (1) edge[bend right] node[above]{0} (0);
            \draw (k-1) edge[bend right] node{0} (2);
            \draw (k) edge[bend right] node{0} (k-1);
            \draw (k+1) edge[bend right] node{0} (k);
            \draw (c) edge[bend right] node{0} (k);

            \draw (0) edge node{0} (-1);
            \draw (-1) edge[bend left] node{0} (-2);
            \draw (-k+1) edge[bend left] node{0} (-k);
            \draw (-k) edge[bend left] node{0} (-k-1);
            \draw (-k-1) edge[bend left] node{0} (-c+1);
            \draw (-c) edge[bend left] node{0} (-c-1);
            \draw (-c-1) edge[loop below] node{0,1} (-c-1);

            \draw (-1) edge[bend left] node{1} (0);
            \draw (-k+1) edge[bend left] node{1} (-2);
            \draw (-k) edge[bend left] node{1} (-k+1);
            \draw (-k-1) edge[bend left] node{1} (-k);
            \draw (-c) edge[bend left] node{1} (-k);
        \end{tikzpicture}
        \caption{Automaton used by \cite{shubert1973testing} to show Theorem \ref{thm:non_adv_sym_bin_det_achievability}.}
        \label{fig:automata_Shubert}
    \end{figure}
    The machine used by \cite{berg2020binary}, upon seeing a zero after a (not too long) sequence of ones, forgets all about the ones and starts counting zeros.
    The machine used by \cite{shubert1973testing}, in the same situation, keeps some memory of the ones it saw.
    Again, it can be easily seen that this automaton satisfies \eqref{eq:automata_regularity}.
\end{proof}

\section{Discussion}
\label{sec:discussion}
Memory limitation is a realistic resource constraint, and it is important to understand its impact on adversarially robust statistical inference. In the context of the problem considered here, many aspects are worth studying, including other contamination models, speed of convergence, and power of randomness.

The results of \cite{hellman1970learning,hellman1971memory,berg2020binary} show that randomness is a useful resource for (non-adversarial) hypothesis testing using finite-memory time-invariant algorithms.
To further delineate its role in the adversarial setting, consider four variants: (i) deterministic FSMs; (ii) randomized FSMs whose randomness (random coins) are revealed to the adversary in advance; (iii) randomized FSMs with private random coins, as considered in this work; and (iv) randomized FSMs with private random coins whose states are not observed by the adversary. In all four variants, the adversary observes past samples; in (i)--(iii), it also observes the history of FSM states. For a fixed number of states, the optimal minimax errors in these models are nonincreasing from (i) to (iv).

Our converse uses only fixed i.i.d. sampling distributions and therefore applies even to model (iv).
Hence, by Remark~\ref{rem:C_gamma_relation}, the error exponents for (iii) and (iv) coincide, and the optimal error with $S+1$ states in (iii) is no larger than that with $S$ states in (iv). It may indeed be possible to show a gap by obtaining a tighter lower bound for (iii), as we do for the $S=2$ case in Appendix~\ref{sec:tight_lower_bound}, and modifying the achievable scheme to take advantage of the adversary's lack of knowledge of the FSM states in (iv) (e.g., using a lazy version of the FSM we use). In general, there is an exponential separation in $S$ between (i) and (iii)~\cite{berg2020binary}; see also the discussion in Section~\ref{sec:deterministic}.

It would be interesting to determine where model (ii) lies between (i) and (iii). One may also consider variants in which, at each time $n$, the adversary knows the random coins that the FSM will use over the next $N$ time steps and can use this information when choosing the distribution of the next sample. Such models are relevant when public randomness, for example from NIST's randomness beacon~\cite{NISTRandomnessBeacon}, is revealed before the adversary chooses the distributions of the corresponding samples. For a fair comparison between (i) and (ii), any memory needed to store random bits for later use might need to be included in the FSM's state budget.

Our analysis concerns the asymptotic time-average probability of error. Obtaining finite-time error bounds~\cite{flower1972hypothesis} and quantifying convergence rates are further directions for study, as is extending the analysis to other contamination models~\cite{diakonikolas2023algorithmic}.

\appendices
\section{Missing proofs from Section \ref{sec:results}}
\label{sec:results_extra}
\subsection{Proof of Theorem \ref{thm:weak_LFD}}
\label{sec:thm2_proof}
\begin{proof}
 By \eqref{eq:gamma-le-gammaHC},
    \begin{align}
        \gamma \leq \min_{\substack{p\in\mathcal{P}_0}, \, q\in\mathcal{P}_1} \gamma_{\mathrm{HC}}(p,q) \leq \gamma_{\mathrm{HC}}(p^*,q^*).\label{eq:LFDcase-converseproof}
    \end{align}
    Now let $f_+=\mathbf{1}_{\mathcal{X}_+}, f_-=\mathbf{1}_{\mathcal{X}_-}$ where $\mathcal{X}_+:= \argmax_{x\in\mathcal{X}} \frac{q^*(x)}{p^*(x)}$, and $\mathcal{X}_-:= \argmin_{x\in\mathcal{X}}\frac{q^*(x)}{p^*(x)}$.
    Then,
    \begin{align*}
        \gamma\ge{}& C,\\
        \ge{}& \min_{\substack{q,q'\in\mathcal{P}_1 \\ p,p'\in\mathcal{P}_0}} \frac{\mathbb{E}_q[f_+]\mathbb{E}_p[f_-]}{\mathbb{E}_{q'}[f_-]\mathbb{E}_{p'}[f_+]}, \\
        ={}& \min_{\substack{q,q'\in\mathcal{P}_1 \\ p,p'\in\mathcal{P}_0}} \frac{q(\mathcal{X}_+)p(\mathcal{X}_-)}{q'(\mathcal{X}_-) p'(\mathcal{X}_+)}, \\
        ={}& \frac{\min_{q\in\mathcal{P}_1}q(\mathcal{X}_+) \min_{p\in\mathcal{P}_0}p(\mathcal{X}_-)}{\max_{q\in\mathcal{P}_1}q(\mathcal{X}_-) \max_{p\in\mathcal{P}_0}p(\mathcal{X}_+)}, \\
        ={}& \frac{q^*(\mathcal{X}_+)p^*(\mathcal{X}_-)}{q^*(\mathcal{X}_-)p^*(\mathcal{X}_+)} = \gamma_\mathrm{HC}(p^*,q^*),
    \end{align*}
    where the last line follows from the second last one via \eqref{eq:weak_LFD_def_0} and \eqref{eq:weak_LFD_def_1}. Together with \eqref{eq:LFDcase-converseproof}, this implies that, whenever a weak LFD pair $(p^*,q^*)$ exists, $\gamma=C=\gamma_\mathrm{HC}(p^*,q^*)$, and hence, the upper and lower bounds in Theorem~\ref{thm:error_characterisation} coincide to give \eqref{eq:weak_LFD_achivability}.
\end{proof}

\subsection{Tight result for $S=2$}
\label{sec:tight_lower_bound}
\begin{thm}
    For $S=2$,
    \begin{align}
        P_e^*(2)=\frac{1}{1+\sqrt{C}}.
    \end{align}
\end{thm}
Hence, for $S=2$, the upper bound in \eqref{eq:upper_bound} of Theorem~\ref{thm:error_characterisation} is tight.
\begin{proof}
    The upper bound follows from \eqref{eq:upper_bound}.
    We show the lower bound.
    Since $S=2$, the machine must have one state declaring $0$ and the other state declaring $1$; otherwise, $P_e=1$.
    Relabel the states as 0 and 1 so that $\hat{H}(s)=s,s\in\{0,1\}$.
    Define
    \begin{align}
        f_+(x):=\pi(1|0,x),\qquad f_-(x):=\pi(0|1,x), \; x\in\mathcal{X},
    \end{align}
    Choose $p^*_0,p^*_1\in\mathcal{P}_0$ and $q^*_0,q^*_1\in\mathcal{P}_1$ as
    \begin{align}
        p^*_0\in{}&\argmax_{p\in\mathcal{P}_0} \mathbb{E}_p[f_+], & p^*_1\in{}&\argmin_{p\in\mathcal{P}_0} \mathbb{E}_p[f_-],\\
        q^*_0\in{}& \argmin_{q\in\mathcal{P}_1} \mathbb{E}_q[f_+], & q^*_1\in{}& \argmax_{q\in\mathcal{P}_1} \mathbb{E}_q[f_-].
    \end{align}    
    The (state-dependent) adversarial strategy is as follows:
    Under $H=0$, choose $p_n(R_{n-1})=p^*_{R_{n-1}}$; under $H=1$, choose $p_n(R_{n-1})=q^*_{R_{n-1}}$.
    First, suppose the resulting Markov chain is irreducible under both hypotheses, with stationary distributions $\mu_0,\mu_1$ under $H=0,1$, respectively.
    Note that 
    \begin{align*}
        \frac{\mathbb{E}_{q^*_0}[f_+]}{\mathbb{E}_{q^*_1}[f_-]} \frac{\mathbb{E}_{p^*_1}[f_-]}{\mathbb{E}_{p^*_0}[f_+]} ={}& C(f_+,f_-)\le \sup_{f_+,f_-}C(f_+,f_-)=C.
    \end{align*}
    where the $\sup$ is over $f_+,f_-:\mathcal{X} \rightarrow[0,1]$.
    The detailed balance gives
    \begin{align*}
        \frac{\mu_0(0)}{\mu_0(1)}={}& \frac{\mathbb{E}_{p^*_1}[f_-]}{\mathbb{E}_{p^*_0}[f_+]}\\
        \frac{\mu_1(1)}{\mu_1(0)}={}& \frac{\mathbb{E}_{q^*_0}[f_+]}{\mathbb{E}_{q^*_1}[f_-]}\\
        \therefore \frac{\mu_0(1)\mu_1(0)}{(1-\mu_0(1))(1-\mu_1(0))}={}& \frac{\mathbb{E}_{p^*_0}[f_+]}{\mathbb{E}_{p^*_1}[f_-]}\frac{\mathbb{E}_{q^*_1}[f_-]}{\mathbb{E}_{q^*_0}[f_+]}\\
        ={}&\frac{1}{C(f_+,f_-)}\ge\frac{1}{C}
    \end{align*}
    Hence, the probability of error is at least $\max\{\mu_0(1),\mu_1(0)\}\ge \frac{1}{1+\sqrt{C}}$.

    For the non-irreducible cases, we have the following claim:
    If $h$ is the only recurrent state under $H=h$, either $C=\infty$ or $h$ is the only recurrent state under $H=1-h$.
    The cases not covered by this claim are:
    \begin{enumerate}
        \item The machine is disconnected under both hypotheses.
        In this case the error is 1.
        \item for some $h\in\{0,1\}$, only state $1-h$ is recurrent under $H=h$.
        In this case the error is 1.
        \item The machine is disconnected under one hypothesis and irreducible under the other one.
        It can be easily checked that in this case $C=\infty$.
    \end{enumerate}

    To prove the claim:
    Suppose 0 is the only recurrent state under $H=0$ and $C<\infty$.
    Then $\mathbb{E}_{p_0^*}[f_+]=0$, and $\mathbb{E}_{p_1^*}[f_-]>0$. If $\mathbb{E}_{q_0^*}[f_+]>0$ the pair of functions $(f_+,\mathbf{1})$ gives $C=\infty$, where $\mathbf{1}$ is the constant function which takes on value 1. Similarly, if $\mathbb{E}_{q_1^*}[f_-]=0$, then the pair of functions  $(\mathbf{1},f_-)$ gives $C=\infty$. Otherwise, $\mathbb{E}_{q_0^*}[f_+]=0$ and $\mathbb{E}_{q_1^*}[f_-]>0$, so 0 is the only recurrent state under $H=1$. The case of $h=1$ follows similarly.
    Hence, we may conclude that
    \[
    \sup_{\mathcal A}P_e(\pi,\hat H;\mathcal A)
    \ge \frac{1}{1+\sqrt C}
    \quad\text{for every two-state machine}.
    \]   
\end{proof}

\section{Examples}
\subsection{Optimal $f_+$ may not be indicator}
\label{sec:example_nonindicator_extra}
Let $\mathcal{X}=\{a,b,c\}$,
$\mathcal{P}_0=\left\{p\in\Delta(\mathcal{X}):p(a)+2p(b)=\frac{1}{2}\right\} = \left\{ \left(\frac{\lambda}{2},\frac{1-\lambda}{4},\frac{3-\lambda}{4}\right):\lambda\in[0,1]\right\}$,
$\mathcal{P}_1=\left\{\left(\frac{1}{3},\frac{1}{3},\frac{1}{3}\right)\right\}$.

\begin{lm}
    \label{lm:eg_achievability}
    For this example $\gamma=4$. The set of $(f_+,f_-)$ which obtain the optimal value in \eqref{eq:gamma_def} is given by $f_+=\alpha\left(\frac{1}{4},\frac{3}{4},0\right)$, $f_-=\beta\left(0,0,1\right)$, where $\alpha\in(0,4/3]$ and $\beta\in(0,1]$.
\end{lm}

\begin{proof}[Proof of Lemma \ref{lm:eg_achievability}]
    Let $p_0=\left(0,\frac{1}{4},\frac{3}{4}\right)$, and $p_1=\left(\frac{1}{2},0,\frac{1}{2} \right)$.
    Then, $\mathcal{P}_0$ is the convex hull of $\{p_0,p_1\}$.
    Since $\gamma(f_+,f_-)$ is invariant under scalings of $f_+$ and those of $f_-$, we set $\sum_{x\in\mathcal{X}}f_+(x)=\sum_{x\in\mathcal{X}}f_-(x)=1$ without loss of generality.
    Then,
    \begin{align}
        \gamma&= \sup_{f_+,f_-\in\Delta(\mathcal{X})}  \min_{\substack{p\in\mathcal{P}_0 \\ q\in\mathcal{P}_1}} \gamma(f_+,f_-;p,q), \label{eq:example_opt_prob}\\
        \text{where } \gamma(f_+,f_-;p,q) &=  \frac{\mathbb{E}_q[f_+]\mathbb{E}_p[f_-]}{\mathbb{E}_q[f_-]\mathbb{E}_p[f_+]}. \nonumber
    \end{align}
    We claim that $f_+^*:=\left(\frac{1}{4},\frac{3}{4}, 0 \right), f_-^*:=(0,0,1); p^*:=\left(\frac{1}{6},\frac{1}{6}, \frac{2}{3}\right), q^*: = \left( \frac{1}{3}, \frac{1}{3},\frac{1}{3}\right)$ is a saddle point of \eqref{eq:example_opt_prob}.
    Note that
    \begin{align*}
        \frac{ \mathbb{E}_{q^*}[f_+^*] }{ \mathbb{E}_{q^*}[f_-^*] }={}&1;\\
        \frac{\mathbb{E}_{p_\lambda}[f_-^*]}{\mathbb{E}_{p_\lambda}[f_+^*]}={}& \frac{(3-\lambda)/4}{\lambda/8+(3-3\lambda)/16}\\
        ={}& 4
    \end{align*}
    Thus,
    \begin{align*}
        \min_{\substack{p\in\mathcal{P}_0 \\ q\in\mathcal{P}_1}}\gamma(f_+^*,f_-^*;p,q)=4.
    \end{align*}
    Now,
    \begin{align}
        \frac{\mathbb{E}_{q^*}[f_+]}{\mathbb{E}_{p^*}[f_+]} ={}& \frac{f_+(a)+f_+(b)+f_+(c)}{f_+(a)/2 + f_+(b)/2 + 2f_+(c)}\le 2 \label{eq:optimal_f_+} \\
        \frac{\mathbb{E}_{p^*}[f_-]}{\mathbb{E}_{q^*}[f_-]} ={}& \frac{f_-(a)/2+f_-(b)/2+2f_-(c)}{f_-(a)+f_-(b)+f_-(c)}\le 2 \label{eq:optimal_f_-}
    \end{align}
    Thus,
    \begin{align*}
        \max_{f_+,f_-\in\Delta(\mathcal{X})} \gamma(f_+,f_-; p^*,q^*)=4
    \end{align*}
    This proves the claim.

    We have shown that the optimal value is $\gamma=4$, and $(f_+^*,f_-^*)$ achieves this value.
    It can be seen that the optimiser $(f_+^*,f_-^*)$ is unique:
    For any $(f_+',f_-')$ which optimises \eqref{eq:example_opt_prob}, due to the tightness of \eqref{eq:optimal_f_+}-\eqref{eq:optimal_f_-}, $f_+'$ has to be supported on a subset of $\{a,b\}$; and $f_-'$ on $\{c\}$.
    Thus, $f_-'=(0,0,1)$.
    Now,
    \begin{align*}
        \gamma(f_+',f_-';p_1,q^*)={}&1 \cdot \frac{\mathbb{E}_{p_1}[f_-']}{\mathbb{E}_{p_1}[f_+']}\\
        ={}&\frac{f_-'(c)/2}{f_+'(a)/2}\\
        ={}& \frac{1}{f_+'(a)}
    \end{align*}
    By optimality of $(f_+',f_-')$, this last expression has to be at least 4, which implies $f_+'(a)\le1/4$.
    Similarly,
    \begin{align*}
        \gamma(f_+',f_-';p_0,q^*)={}&1 \cdot \frac{\mathbb{E}_{p_0}[f_-']}{\mathbb{E}_{p_0}[f_+']}\\
        ={}&\frac{3/4}{(1/4)f_+'(b)}\\
        ={}& \frac{3}{f_+'(b)}
    \end{align*}
    Then, optimality of $(f_+',f_-')$ implies $f_+'(b) \le3/4$.
    Hence $(f_+',f_-')=(f_+^*,f_-^*)$.
\end{proof}

\subsection{$C<\gamma$}
\label{sec:example_C_less_than_gamma_extra}
\begin{lm}
    For the example from Section
    \ref{sec:example_nonindicator_extra} (i.e. $\mathcal{X}=\{a,b,c\}$,
    $\mathcal{P}_0=\left\{p \in \Delta(\mathcal{X}) : p(a)+2p(b) =\frac{1}{2}\right\} =\\ \left\{ \left(\frac{\lambda}{2},\frac{1-\lambda}{4},\frac{3-\lambda}{4}\right):\lambda\in[0,1]\right\}$,
    $\mathcal{P}_1=\left\{\left(\frac{1}{3},\frac{1}{3},\frac{1}{3}\right)\right\}$), we have $3=C<\gamma=4$.
\end{lm}
\begin{proof}
    We show that for this example, $f_+^*:=\left(\frac{1}{3},\frac{2}{3},0\right), f_-^*:=(0,0,1)$ an optimiser of \eqref{eq:C_def}.
    For this $(f_+^*,f_-^*)$,
    \begin{align*}
        C(f_+^*,f_-^*)={}& \min_{p\in\mathcal{P}_0,\; q\in\mathcal{P}_1}\frac{\mathbb{E}_q[f_+^*]}{\mathbb{E}_p[f_+^*]} \min_{p\in\mathcal{P}_0,\; q\in\mathcal{P}_1}\frac{\mathbb{E}_p[f_-^*]}{\mathbb{E}_q[f_-^*]},\\
        ={}& \left(\min_{\lambda\in[0,1]}\frac{1}{\frac{1}{3}\frac{\lambda}{2}+\frac{2}{3}\frac{1-\lambda}{4}}\right)\cdot\left(\min_{\lambda\in[0,1]}\frac{3-\lambda}{4}\right)\\
        ={}&3.
    \end{align*}

    Now suppose $(f_+',f_-')\in\Delta(\mathcal{X})\times \Delta(\mathcal{X})$ is an optimiser for \eqref{eq:C_def}.
    \begin{align*}
        C(f_+',f_-')={}& \frac{\min_{p\in\mathcal{P}_0,\; q\in\mathcal{P}_1}\frac{\mathbb{E}_p[f_-']}{\mathbb{E}_q[f_-']}}{\max_{p\in\mathcal{P}_0,\; q\in\mathcal{P}_1}\frac{\mathbb{E}_p[f_+']}{\mathbb{E}_q[f_+']}}\\
        ={}& \frac{\min_{p\in\mathcal{P}_0}\mathbb{E}_p[f_-']}{\max_{p\in\mathcal{P}_0}\mathbb{E}[f_+']}
    \end{align*}
    Note that
    \begin{align}
        \min_{p\in\mathcal{P}_0}\mathbb{E}_p[f_-'] \le \mathbb{E}_{p_1}[f_-']={}&\frac{1}{2}(f_-'(a)+f_-'(c)) \le \frac{1}{2}.\\
        \intertext{Also,}
        \max_{p\in\mathcal{P}_0}\mathbb{E}_p[f_+']\ge \mathbb{E}_{p_{1/3}}[f_+']={}& \frac{1}{6}(f_+'(a)+f_+'(b))+\frac{2}{3}f_+'(c)\ge\frac{1}{6}.
    \end{align}
    Thus, $C\le 3$.

    An optimal $f_+',f_-'$ must satisfy
    \begin{align*}
        \frac{1}{6}(f_+'(a)+f_+'(b))+\frac{2}{3}f_+'(c)=\frac{1}{6}
    \end{align*}
    Thus, $f_+'(c)=0$.
    Furthermore we must have
    \begin{align*}
        \mathbb{E}_{p_0}[f_+']=\mathbb{E}_{p_1}[f_+']=\frac{1}{6}
    \end{align*}
    forcing $f_+=\left(\frac{1}{3},\frac{2}{3},0\right)$.
\end{proof}
Note also that the optimisers for \eqref{eq:C_def} and \eqref{eq:gamma_def} are different.
When designing an automaton for this problem, the edge transition probability $f_+^C$ would need to be different from the internal transition probabilities $f_+^\gamma$.
\subsection{Weak LFD pairs may exist even when strong LFD pairs do not}
\label{sec:example_weak_vs_strong_extra}
Let $\mathcal{X}=\{a,b,c,d,e\}$, $\mathcal{P}_1=\{q\}=\left\{\left(\frac{9}{50}, \frac{4}{15}, \frac{4}{15}, \frac{4}{15}, \frac{1}{50} \right)\right\}$,
$\mathcal{P}_0=\{p_\lambda:\lambda\in[0,1]\}=\left\{\left(\frac{1}{50}, \frac{2\lambda}{5},\frac{1-\lambda}{5},\frac{3-\lambda}{5},\frac{9}{50} \right):\lambda\in[0,1]\right\}$\\
$ = \mathrm{conv}\left(\left\{\left(\frac{1}{50}, 0, \frac{1}{5}, \frac{3}{5}, \frac{9}{50} \right)\,,\,\left(\frac{1}{50}, \frac{2}{5}, 0, \frac{2}{5}, \frac{9}{50} \right)\right\}\right)$.
We now show that, for this example, a weak LFD pair exists, but no (strong) LFD pair exists.

First, we show that $(p_\lambda,q)$ is a weak LFD pair for every $\lambda\in\left(\frac{2}{27},\frac{23}{27}\right)$.
To see this, first note that the likelihood ratio of $q/p_\lambda$ is
\begin{align*}
    \frac{q(x)}{p_\lambda(x)}=\left(9,\frac{2}{3\lambda}, \frac{4}{3(1-\lambda)}, \frac{4}{3(3-\lambda)}, \frac{1}{9} \right).
\end{align*}
For $\lambda\in\left(\frac{2}{27},\frac{23}{27}\right)$,
\begin{align*}
    \max_x\frac{q(x)}{p_\lambda(x)}={}&9, & \mathcal{X}_+=\argmax_x\frac{q(x)}{p_\lambda(x)}={}&\{a\},\\
    \min_x\frac{q(x)}{p_\lambda(x)}={}&\frac{1}{9}, & \mathcal{X}_-=\argmin_x\frac{q(x)}{p_\lambda(x)}={}&\{e\}.
\end{align*}
It is easy to see that \eqref{eq:weak_LFD_def_0} is satisfied for this $(\mathcal{X}_+,\mathcal{X}_-)$ (and \eqref{eq:weak_LFD_def_1} is satisfied trivially).
On the other hand, no pair $(p_\lambda,q)$ forms a strong LFD pair:
\paragraph{Case I: $\lambda=1$}
In this case $p_\lambda(c)=0$, 
so $\frac{q(c)}{p_\lambda(c)}=\infty$.
Consider equation \eqref{eq:LFD_pair_def_0} with $\eta=10$. Note that $\{x:q(x)/p_\lambda(x)>10\}=\{c\}$.
But clearly, there are $p\in\mathcal{P}_0$ with $p(c)>p_\lambda(c)=0$, so \eqref{eq:LFD_pair_def_0} is unsatisfied for $\eta=10$
\paragraph{Case II: $\lambda<1$}
In this case consider equation \eqref{eq:LFD_pair_def_0} with $\eta=\frac{2}{3}$ 
Now, $\{x:q(x)/p_\lambda(x)>2/3\}=\{a,b,c\}$, and $p_\lambda(\{a,b,c\})=\frac{1}{50}+\frac{\lambda+1}{5}$. But, since $\lambda<1$, there is another $\lambda'\in(\lambda,1)$ with $p_{\lambda'}(\{a,b,c\})=\frac{1}{50}+\frac{\lambda'+1}{5}>p_\lambda(\{a,b,c\})$.
Hence, \eqref{eq:LFD_pair_def_0} is unsatisfied for $\eta=\frac{2}{3}$.

\section*{AI Disclosure}

We worked on the problem and produced an initial draft (extended version for our ITW 2026 submission), which is available at \url{https://arxiv.org/pdf/2605.12063v1}, without taking any help from AI. Later, we used AI (ChatGPT-6 Astra) to obtain a review. Based on the review, we made the following substantive changes
:
\begin{itemize}
    \item Our first draft used a martingale argument to analyze the FSM we use to show the achievability part of Theorem~\ref{thm:error_characterisation} (in Section~\ref{sec:achievability}). The review suggested using an occupation measure analysis. In our submission, we have updated the analysis in our achievability proof along these lines. The new analysis itself is written manually.
    \item For the example in Section~\ref{sec:example_nonindicator}, the review suggested a shorter proof for deriving the optimal $f_+,f_-$ than what we had. We have adopted this proof idea.
\end{itemize}
In preparing the submitted version, we also used AI to help with minor copy-editing and improve clarity/grammar. Apart from the above, the additional results relative to the conference version were developed without help from AI.

To summarize, help from AI has materially affected the achievability proof we present in Section~\ref{sec:achievability} and the proof in Appendix~\ref{sec:example_nonindicator_extra}.

\bibliographystyle{IEEEtran}
\bibliography{bib}
\end{NoHyper}
\end{document}